\documentclass{amsart}
\usepackage[utf8]{inputenc}
\usepackage[T1]{fontenc}
\usepackage{amsmath,amssymb,amsthm}
\usepackage{pdfpages}
\usepackage{tikz-cd}
\usepackage{color}
\usepackage[colorlinks=true]{hyperref}
\hypersetup{urlcolor=blue, citecolor=red}

\usepackage[backend=bibtex,giveninits=true,maxnames=5,doi=true,isbn=false,url=false,sortcites,style=phys,biblabel=brackets,sorting=nyt]{biblatex} 

\AtEveryBibitem{%
    \clearlist{language}
}

\newcommand{\F}{\mathbb{F}}

\newcommand*{\RR}{\mathbb{R}}

\newcommand{\HorD}{\mathcal{H}}
\newcommand{\VertD}{\mathcal{V}}
\newcommand{\pHor}{\pi_\HorD}
\newcommand{\pVert}{\pi_\VertD}
\newcommand{\Reeb}{\mathcal{R}}

\newcommand{\lsharp}{\sharp_\Lambda}

\newcommand{\dhbot}{\bot_{d\eta}}

\numberwithin{equation}{section}

\newtheorem{theorem}{Theorem}
\newtheorem{proposition}{Proposition}
\newtheorem{corollary}{Corollary}
\newtheorem{lemma}{Lemma}
\theoremstyle{definition}
\newtheorem{definition}{Definition}

\theoremstyle{remark}
\newtheorem{remark}{Remark}
\newtheorem{example}{Example}

\newcommand{\im}{\mathrm{im}}
\newcommand{\dv}[2]{\frac{d#1}{d#2}}
\newcommand{\pdv}[2]{\frac{\partial #1}{\partial #2}}
\newcommand{\ann}[1]{{#1}^\circ}
\newcommand{\pr}{\mathrm{pr}}
\newcommand*{\orth}[1]{{#1}^{\bot}}
\newcommand*{\orthL}[1]{{}^{\bot}{#1}}
\newcommand*{\lieD}[1]{\mathcal{L}_{#1}}
\newcommand*{\NHBr}[1]{\{#1\}_{L,\Delta}}

\usepackage{color}
\usepackage[all]{xy}
\usepackage{graphicx}%

\title[A review on contact systems]
      {A review on contact Hamiltonian and Lagrangian systems}

\author[M. de Le\'on and M. Lainz]{}

\email{mdeleon@icmat.es,  manuel.lainz@icmat.es}


\thanks{\noindent We acknowledge the financial support from the MINECO Grant MTM2016-76-072-P, and the ICMAT Severo Ochoa projects SEV-2011-0087 and SEV-2015-0554. Manuel Laínz wishes to thank MICINN and ICMAT for a FPI-Severo Ochoa predoctoral contract PRE2018-083203.}

\begin{document}

\maketitle

\centerline{\scshape Manuel de Le\'on}
\medskip
{\footnotesize
 \centerline{C/ Nicol\'as Cabrera, 13--15, 28049, Madrid. SPAIN }
   \centerline{Instituto de Ciencias Matemáticas and Real Academia Espa\text{\~{n}}ola de Ciencias}
 } %

\medskip

\centerline{\scshape Manuel Lainz}
\medskip
{\footnotesize
 \centerline{C/ Nicol\'as Cabrera, 13--15, 28049, Madrid. SPAIN }
   \centerline{Instituto de Ciencias Matemáticas}
 } %

\begin{abstract}
Contact Hamiltonian dynamics is a subject that
has still a short history, but with relevant applications in many areas: thermodynamics, cosmology, 
control theory, and neurogeometry, among others. In recent years there has been a great effort to study this type of dynamics both in theoretical 
aspects and in its potential applications in geometric mechanics and mathematical physics. 
This paper is intended to be a review of some of the results that the authors and their collaborators have recently obtained on the subject.
\end{abstract}

\tableofcontents

\section{Introduction}

Contact geometry is a topic of great interest in Differential Geometry, but it has been revealed
relevant in the last years due to its applications to describe mechanical dissipative systems, both in the Hamiltonian and Lagrangian descriptions. Some of these applications are in thermodynamics (both reversible~\cite{Mrugala1991}, and, more recently, irreversible~\cite{Grmela2014,Eberard2007,Simoes2020b}), statistical mechanics~\cite{Bravetti2016},
control theory~\cite{deLeon2020a}, neurogeometry, economics, cosmology, among others.

Indeed, contact Hamiltonian mechanics is related with the work of G. Herglotz~\cite{Herglotz1930} almost 90 years ago, who used a generalization of the well-known Hamilton principle (that includes to solve an implicit differential equation before to define the action) that, almost miraculously, provides the same equations that we can obtain using contact geometry~\cite{Georgieva2011,deLeon2019a}. 

In the Hamiltonian picture, the setting is the extended cotangent bundle $T^*Q \times {\mathbb R}$ equipped with
its canonical contact form $\eta = dz - p_i dq^i$, where $(q^i, p_i, z)$ are bundle coordinates. Given
a Hamiltonian function $H : T^*Q \times {\mathbb R} \to {\mathbb R} $ the contact Hamilton equations are

\begin{eqnarray}
\frac{dq^i}{dt} & = & \frac{\partial H}{\partial p_i}, \\
\frac{dp_i}{dt} & = & - (\frac{\partial H}{\partial q^i} +  p_i \frac{\partial H}{\partial z}),\\
\frac{dz}{dt} & = & (p_i \frac{\partial H}{\partial p_i} - H).
\end{eqnarray}

On the other hand, given a Lagrangian function $L : TQ \times {\mathbb R} \to {\mathbb R}$
on the extended tangent bundle $T^*Q \times {\mathbb R}$ we obtain (using the Herglotz principle) the
contact Lagrangian equations
\begin{equation}
\frac{d}{dt} (\frac{\partial L}{\partial \dot{q}^i}) - \frac{\partial L}{\partial q^i} =
\frac{\partial L}{\partial \dot{q}^i} \frac{\partial L}{\partial z},
\end{equation}
where $(q^i, \dot{q}^i, z)$ are bundle coordinates. Of course, both equations are related
through the Legendre transform (we will assume that $L$ is regular).

In this paper we present a survey on some of the recent developments on
contact Hamiltonian and Lagrangian mechanics of our group of research. It is not an exhaustive
account of all these results, but only some of them. So, after introducing the main aspects of contact
Hamiltonian systems (Section 2), and contact Lagrangian systems (Section 3), and the common description
of symplectic and contact structures
under the framework of Jacobi structures (Section 4), we
discuss the following subjects:

\begin{enumerate}

\item The role of submanifolds in contact Hamiltonian systems and the interpretation of its
dynamics as Legendrian submanifolds. A coisotropic reduction theorem is also introduced.

\item The extension of the notion of momentum map to this scenario as well as the corresponding reduction theorem,
in the same vein as in the symplectic case.

\item A relevant subject in dynamics is the relation between symmetries and conserved quantities via the different
generalizations of Noether theorem. We extend the well-known results in symplectic mechanics to contact dynamics,
but in the latter case we obtain dissipated quantities instead of conserved ones.

\item The Hamilton-Jacobi theory is also explored.

\item We also consider the case of singular Lagrangian systems, and obtain
a constraint algorithm that provides a Jacobi bracket on the final constraint
submanifold, that we call Dirac-Jacobi bracket.

\item A new subject is the contact description of nonholonomic mechanical systems,
that allows us to consider such systems when some kind of dissipation is considered.
The corresponding nonholonomic bracket is constructed (indeed, it is an almost Jacobi bracket).

\end{enumerate}

Finally, we list a series of subjects that have been also studied in these last two years as well as others
thar are being now investigated.

\section{Contact Hamiltonian systems}

In this section we will recall the three main geometric structures~\cite{deLeon2017} involved in the description of Hamiltonian dynamics.

\subsection{Symplectic Hamiltonian systems}

As it is well known, Hamiltonian dynamics are developed using symplectic geometry~\cite{Arnold1997,deLeon2011,Abraham1978}. Indeed, let $(M, \omega)$ be a symplectic manifold, that is, 
$\omega$ is a non--degenerate closed 2-form, say $d \omega = 0$ and $\omega^n \not= 0$, where $M$ has even dimension $2n$.
Then, if $H : M \to \RR$ is a Hamiltonian function, the Hamiltonian vector field $X_H$
is obtained using the equation
\begin{equation}\label{hsymp}
\flat (X_H) = dH,
\end{equation}
where $\flat$ is the vector bundle isomorphism
$$
\flat : TM \to T^* M \; , \; \flat(v) = i_v \, \omega.
$$
In Darboux coordinates $(q^i, p_i)$ we have
$\omega = dq^i \wedge dp_i$ and
$$
X_ H = \frac{\partial H}{\partial p_i} \frac{\partial}{\partial q^i} - 
\frac{\partial H}{\partial q^i} \frac{\partial}{\partial p_i}.
$$
In such a way that an integral curve $(q^i(t), p_i(t))$ satisfies the Hamilton equations
\begin{equation}\label{hsympl2}
\frac{dq^i}{dt} = \frac{\partial H}{\partial p_i} \; , \; 
\frac{dp_i}{dt} = - \frac{\partial H}{\partial q^i}.
\end{equation}

\subsection{Cosymplectic Hamiltonian systems}\label{sec:cosymplectic_hamiltonian_systems}

A cosymplectic structure on an odd-dimensional manifold~\cite{deLeon2017,Cantrijn1992} $M$ is a pair $(\Omega, \eta)$ where
$\Omega$ is a closed 2-form, $\eta$ is a closed 1-form, and
$\eta \wedge \Omega^n \not= 0$; here, $M$ has dimension $2n+1$.
$(M, \Omega, \eta)$ will be called a cosymplectic manifold.

There is a Darboux theorem for a cosymplectic manifold, that is,
there are local coordinates (called Darboux coordinates) $(q^i, p_i, z)$ around any point of $M$
such that
$$
\Omega = dq^i \wedge dp_i \; , \; \eta = dz.
$$
There also exists a unique vector field (called Reeb vector field) $\mathcal R$ such that
$$
i_{\mathcal R} \, \Omega = 0 \; , \; i_{\mathcal R}\, \eta = 1.
$$
In Darboux coordinates we have
$$
\mathcal R = \frac{\partial}{\partial z}
$$

Let $H : M \to \RR$ be a Hamiltonian function, say $H = H(q^i, p_i, z)$.

Consider the vector bundle isomorphism
$$
\tilde{\flat} : TM \to T^* M \; , \; \flat(v) = i_v \, \Omega + \eta (v) \, \eta
$$
and define the gradient of $H$ by
$$
\tilde{\flat}({\rm grad} \; H) = dH.
$$
Then
\begin{equation}\label{hcosymp}
{\rm grad} \; H = \frac{\partial H}{\partial p_i} \frac{\partial}{\partial q^i} - 
\frac{\partial H}{\partial q^i} \frac{\partial}{\partial p_i} + \frac{\partial H}{\partial z} \, \frac{\partial}{\partial z}.
\end{equation}

Next we can define two more vector fields:

\begin{itemize}

\item The Hamiltonian vector field 
$$
X_H = {\rm grad} \; H - \mathcal R (H) \mathcal R \; ,
$$

\item and the evolution vector field
$$
{\mathcal E}_H = X_H + {\mathcal R}.
$$
\end{itemize}

From (\ref{hcosymp}) we obtain the local expression

\begin{equation}\label{hcosymp2}
{\mathcal E}_H = \frac{\partial H}{\partial p_i} \frac{\partial}{\partial q^i} - 
\frac{\partial H}{\partial q^i} \frac{\partial}{\partial p_i} +  \frac{\partial}{\partial z}.
\end{equation}
Therefore, an integral curve $(q^i(t), p_i(t), z(t))$ of ${\mathcal E}_H$ satisfies the 
time-dependent Hamilton equations

\begin{eqnarray}\label{hsympl3}
\frac{dq^i}{dt} & = & \frac{\partial H}{\partial p_i}, \\
\frac{dp_i}{dt} & = & - \frac{\partial H}{\partial q^i},\\
\frac{dz}{dt} & = & 1,
\end{eqnarray}
and then $z=t+const$ so that both coordinates can be identified.

\subsection{Contact Hamiltonian systems}\label{sec:contact_hamiltonian_systems}

Consider now a contact manifold~\cite{deLeon2017,deLeon2019a,Bravetti2017} $(M, \eta)$ with contact form $\eta$; this means that
$\eta \wedge d\eta^n \not= 0$ and $M$ has odd dimension $2n+1$.
Again there exists a unique vector field $\mathcal R$ (also called Reeb vector field) such that
$$
i_{\mathcal R} \, d\eta = 0 \; , \; i_{\mathcal R}\, \eta = 1.
$$

There is a Darboux theorem for contact manifolds so that around each point in $M$ one can find local coordinates 
(called Darboux coordinates) $(q^i, p_i, z)$ 
such that
\begin{equation}
    \eta = dz - p_i \, dq^i
\end{equation}
and we have
\begin{equation}
    \mathcal R = \frac{\partial}{\partial z}.
\end{equation}

Define now the vector bundle isomorphism
\begin{equation}\label{eq:contact_iso}
    \bar{\flat} : TM \to T^* M \; , \; \bar{\flat}(v) = i_v \, d\eta + \eta (v) \, \eta
\end{equation}

For a Hamiltonian function $H$ on $M$ we define the Hamiltonian vector field by
\begin{equation}\
    \bar{\flat} (X_H) = dH - (\mathcal R (H) + H) \, \eta
\end{equation}

In Darboux coordinates we get this local expression

\begin{equation}\label{hcont2}
X_H = \frac{\partial H}{\partial p_i} \frac{\partial}{\partial q^i} - 
(\frac{\partial H}{\partial q^i} + p_i \frac{\partial H}{\partial z} \,  \frac{\partial}{\partial p_i} + 
(p_i \frac{\partial H}{\partial p_i} - H) \, \frac{\partial}{\partial z}.
\end{equation}
Therefore, an integral curve $(q^i(t), p_i(t), z(t))$ of $X_H$ satisfies the 
dissipative Hamilton equations

\begin{eqnarray}\label{hcont3}
\frac{dq^i}{dt} & = & \frac{\partial H}{\partial p_i}, \\
\frac{dp_i}{dt} & = & - (\frac{\partial H}{\partial q^i} +  p_i \frac{\partial H}{\partial z}),\\
\frac{dz}{dt} & = & (p_i \frac{\partial H}{\partial p_i} - H).
\end{eqnarray}

\begin{remark}{
Let us say some words to the term dissipative used in this paper. Consider a Hamiltonian system given by the Hamiltonian

$$
H(q,p,z) = \frac{p^2}{2m} \, + V(q) + \gamma \, z
$$

\noindent where $\gamma$ is a constant. This Hamiltonian corresponds to a system with a friction force that depends linearly on the velocity (in our case, on the momenta).

If we apply the contact Hamiltonian mechanism, we obtain the following dynamical equations
\begin{eqnarray*}
&&\dot{q} = \frac{p}{m},\\
&&\dot{p} = - \frac{\partial V}{\partial q} - \gamma \, z,\\
&&\dot{z} = \frac{p^2}{2m} \, - V(q) - \gamma \ z,
\end{eqnarray*}
that are just the damped Newtonian equations. In this sense, dissipation is described by contact Hamiltonian systems, but the theory is even more general.
}
\end{remark}

\begin{remark}{ 
As one can easily see, the contact Hamilton equations are so far to be considered as a simple odd-dimensional counterpart of the symplectic ones. }
\end{remark}

\section{Contact Lagrangian systems}

\subsection{The geometric setting}

Let $L : TQ \times \mathbb{R}\to \RR$ be a Lagrangian function, 
where $Q$ is a configuration $n$-dimensional manifold. Then, $L = L(q^i, \dot{q}^i, z)$, where
$(q^i)$ are coordinates in $Q$, $(q^i, \dot{q}^i)$ are the induced bundle coordinates in $TQ$
and $z$ is a global coordinate in $\RR$.

We will assume that $L$ is regular, that is, the Hessian matrix
$$
\left( \frac{\partial^2 L}{\partial \dot{q}^i \partial \dot{q}^j} \right)
$$
is regular.

From $L$, and using the canonical endomorphism $S$ on $TQ$ locally defined by
$$
S = d q^i \otimes \frac{\partial}{\partial \dot{q}^i}
$$
one can construct a 1-form $\lambda_L$ given by
$$
\lambda_L = S^* (dL),
$$
where now $S$ and $S^*$ are the natural extension of $S$ and its adjoint operator $S^*$ to $TQ \times \RR$.

Therefore, we have
$$
\lambda_L = \frac{\partial L}{\partial \dot{q}^i} \, dq^i
$$
Now, the 1-form
$$
\eta_L = dz -  \frac{\partial L}{\partial \dot{q}^i} \, dq^i .
$$
is a contact form on $TQ \times \RR$ if and only if $L$ is regular; indeed,
if $L$ is regular, then
$$
\eta_L \wedge (d\eta_L)^n \not= 0,
$$ 
and conversely. From now on, we always assume that it is the case. The corresponding Reeb vector field is
$$
{\mathcal R}_L = \frac{\partial}{\partial z} - W^{ij} \frac{\partial^2 L}{\partial \dot{q}^j \partial z} \, \frac{\partial}{\partial \dot{q}^i} ,
$$
where $(W^{ij})$ is the inverse matrix of the Hessian $(W_{ij})$.
The energy of the system is defined by 
$$
E_L = \Delta (L) - L,
$$
where $\Delta = \dot{q}^i \, \frac{\partial}{\partial \dot{q}^i}$ is the Liouville vector field on $TQ$ extended in the usual way to $TQ \times \RR$. Therefore,
$$
E_L = \dot{q}^i \, \frac{\partial L}{\partial \dot{q}^i} - L.
$$

Denote by
$$
\flat_L : T(TQ \times \RR ) \to T^* (TQ \times \RR )
$$
the vector bundle isomorphism
$$
\flat_L (v) = i_v (d\eta_L) + (i_v \eta_L) \, \eta_L
$$
given by the contact form $\eta_L$ on $TQ \times \RR$.
We shall denote its inverse by $\sharp_L = (\flat_L)^{-1}$.

Let $\bar{\xi}_L$ be the unique vector field defined by the equation
\begin{equation}\label{clagrangian1}
\flat_L (\bar{\xi}_L) = dE_L - (\mathcal R_L E_L) + E_L) \, \eta_L.
\end{equation}

A direct computation from eq. (\ref{clagrangian1}) shows that $\bar{\xi}_L$ is locally given by
\begin{equation}\label{clagrangian2}
\bar{\xi}_L = \dot{q}^i \, \frac{\partial}{\partial q^i} + {\mathcal B}^i \,\frac{\partial}{\partial \dot{q}^i} 
+ (L - \dot{q}^i\frac{\partial}{\partial z}(\frac{\partial L}{\partial \dot{q}^i})) \, \frac{\partial}{\partial z},
\end{equation}
where the components ${\mathcal B^i}$ satisfy the equation
\begin{equation}\label{clagrangian3}
{\mathcal B}^i \, \frac{\partial}{\partial \dot{q}^i}(\frac{\partial L}{\partial \dot{q}^j}) 
+ \dot{q}^i \, \frac{\partial}{\partial q^i}(\frac{\partial L}{\partial \dot{q}^j}) 
+ (L - \dot{q}^i\frac{\partial}{\partial z}(\frac{\partial L}{\partial \dot{q}^i})) - \frac{\partial L}{\partial q^i} =
\frac{\partial L}{\partial \dot{q}^i} \frac{\partial L}{\partial z}
\end{equation}

Then, if $(q^i(t), \dot{q}^i(t), z(t))$ is an integral curve of $\bar{\xi}_L$, and substituting its values in eq. (\ref{clagrangian3}) we obtain
$$
{\ddot{q}}^i \, \frac{\partial}{\partial \dot{q}^i}(\frac{\partial L}{\partial \dot{q}^j}) 
+ \dot{q}^i \, \frac{\partial}{\partial q^i}(\frac{\partial L}{\partial \dot{q}^j}) 
+  \dot{z} \frac{\partial}{\partial z}(\frac{\partial L}{\partial \dot{q}^i})) - \frac{\partial L}{\partial q^i} =
\frac{\partial L}{\partial \dot{q}^i} \frac{\partial L}{\partial z},
$$
which corresponds to the generalized Euler-Lagrange equations considered by G. Herglotz in 1930.
\begin{equation}\label{clagrangian4}
\frac{d}{dt} (\frac{\partial L}{\partial \dot{q}^i}) - \frac{\partial L}{\partial q^i} =
\frac{\partial L}{\partial \dot{q}^i} \frac{\partial L}{\partial z}.
\end{equation}

\subsection{Variational formulation of contact Lagrangian mechanics: Herglotz principle}\label{sec:Herglotz_principle}

Let $L:TQ \times \mathbb{R}\to \RR$ be a Lagrangian function. In this subsection we will recall the so-called Herglotz's principle~\cite{Georgieva2003,Herglotz1930,deLeon2019a}, a modification of Hamilton's principle that allows us to obtain Herglotz's equations, sometimes called generalized Euler-Lagrange equations. 

Fix $q_1,q_2 \in Q$ and an interval $[a,b] \subset \RR$. We denote by $\Omega(q_1,q_2, [a,b]) \subseteq({\mathcal{C}}^\infty([a,b]\to Q))$ the space of smooth curves $\xi$ such that $\xi(a)=q_1$ and $\xi(b)=q_2$. This space has the structure of an infinite dimensional smooth manifold whose tangent space at $\xi$ is given by the set of vector fields over $\xi$ that vanish at the endpoints, that is,
\begin{equation}
\begin{aligned}
        T_\xi \Omega(q_1,q_2, [a,b]) =  \{&
            v_\xi \in {\mathcal{C}}^\infty([a,b] \to TQ) \mid \\& 
            \tau_Q \circ v_\xi = \xi, \,v_\xi(a)=0, \, v_\xi(b)=0 
            \}.
\end{aligned}
\end{equation}

We will consider the following maps. Fix $c \in \RR$. Let 
\begin{equation}
    \mathcal{Z}:\Omega(q_1,q_2, [a,b])  \to {\mathcal{C}}^\infty ([a,b] \to \RR )
\end{equation}
 be the operator that assigns to each curve $\xi$ the curve $\mathcal{Z}(\xi)$ that solves the following ODE:
\begin{equation}\label{contact_var_ode}
    \frac{d \mathcal{Z}(\xi)(t)}{dt} = L(\xi(t), \dot \xi(t), \mathcal{Z}(\xi)(t)), \quad \mathcal{Z}(\xi)(a)= c.
\end{equation}

Now we define the \emph{action functional} as the map which assigns to each curve the solution to the previous ODE evaluated at the endpoint:
\begin{equation}\label{contact_action}
    \begin{aligned}
        \mathcal{A}: \Omega(q_1,q_2, [a,b]) &\to \RR ,\\
        \xi &\mapsto \mathcal{Z}(\xi)(b),
    \end{aligned}
\end{equation}
that is, $\mathcal{A} = ev_b \circ \mathcal{Z}$, where $ev_b: \zeta \mapsto \zeta(b)$ is the evaluation map at $b$.

\begin{theorem}(Contact variational principle)
    Let $L: TQ \times \mathbb{R}\to \RR$ be a Lagrangian function and let $\xi\in  \Omega(q_1,q_2, [a,b])$ be a curve in $Q$. Then, $(\xi,\dot\xi, \mathcal{Z}(\xi))$ satisfies Herglotz's equations  if and only if $\xi$ is a critical point of $\mathcal{A}$.
\end{theorem}

\begin{remark}{\rm
    This theorem generalizes Hamilton's Variational Principle. In the case that the Lagrangian is independent of the $\RR$ coordinate, i.e., $L(x,y,z) \linebreak = \hat{L}(x,y)$), and then  
the contact Lagrange equations reduce to the usual Euler-Lagrange equations. In this situation, we can integrate the ODE of (\ref{contact_action}) and we get
    \begin{equation}
        \mathcal{A}(\xi) = \int_a^b \hat L(\xi(t),\dot\xi(t))dt + \frac{c}{b-a},
    \end{equation}
    that is, the usual Euler-Lagrange action up to a constant.}
\end{remark}

\begin{remark}{ 

We will recall here the geometric formalism for time-dependent Lagrangian systems, just to show the differences with the previous contact formalism.
In this case, we also have a regular Lagrangian $L : TQ \times \mathbb{R}\to \RR$,
but instead to consider the contact 1-form $\eta_L$ we will consider the cosymplectic structure given by the pair
$(\Omega_L, dz)$, where
$$
\Omega_L = - d \lambda_L .
$$
It is easy to check that, indeed, if $L$ is regular then
$$
dz \wedge \Omega_L^n \not= 0,
$$
and conversely. Again, we have a Reeb vector field 
$$
\mathcal R_L = \frac{\partial}{\partial z} - W^{ij} \frac{\partial^2 L}{\partial \dot{q}^j \partial z} \, \frac{\partial}{\partial \dot{q}^i}.
$$

Consider now the following vector fields determined by means of the vector bundle isomorphism
\begin{eqnarray*}
&&\widetilde{\flat_L}  :  T(TQ \times \RR ) \to T^*(TQ \times \RR )\\
&&\widetilde{\flat_L} (v) = i_v \, \Omega_L + dz (v) \, dz
\end{eqnarray*}
say,
\begin{enumerate}
\item the gradient vector field
$$
{\rm grad} \; (E_L) = \widetilde{\sharp_L} (dE_L),
$$
\item the Hamiltonian vector field
$$
X_{E_L} = {\mathcal E}_L - \mathcal R (E_L) \, \mathcal R_L,
$$
\item and the evolution vector field
$$
{\mathcal E}_L = X_{E_L} + \mathcal R_L,
$$
\end{enumerate}
where $\widetilde{\sharp_L} = (\widetilde{\flat_L})^{-1}$ is the inverse of $\widetilde{\flat_L}$.

The evolution vector field ${\mathcal E}_L$ is locally given by
\begin{equation}\label{cosylagr1}
{\mathcal E}_ L = \dot{q}^i \, \frac{\partial}{\partial q^i} + B^i \, \frac{\partial}{\partial \dot{q}^i} + \frac{\partial}{\partial z},
\end{equation}
where
\begin{equation}\label{cosylagr2}
B^i \, \frac{\partial}{\partial \dot{q}^i}(\frac{\partial L}{\partial \dot{q}^j}) + \dot{q}^i \, 
\frac{\partial}{\partial q^i}(\frac{\partial L}{\partial \dot{q}^j}) - \frac{\partial L}{\partial q^j} = 0.
\end{equation}
Now, if $(q^i(t), \dot{q}^i(t), z(t))$ is an integral curve of ${\mathcal E}_L$ then it satisfies the 
usual Euler-Lagrange equations
\begin{equation}\label{cosylagr3}
\frac{d}{dt} \left(\frac{\partial L}{\partial \dot{q}^i} \right) - \frac{\partial L}{\partial q^i} = 0,
\end{equation}
since $z = t + constant$. 
}
\end{remark}

\subsection{The Legendre transformation and the Hamiltonian counterpart}

\subsubsection{The classical Hamiltonian geometric setting}

Let $H : T^*Q \times \RR \to \RR$ be a Hamiltonian function, say
$H = H(q^i, p_i, z)$ where $(q^i, p_i, z)$  are bundle coordinates in $T^*Q \times \RR$.
Consider the 1-form
$$
\eta = dz - \theta_Q,
$$
where $\theta_Q$ is the canonical Liouville form on $T^*Q$ and we are considering
the usual identifications for a form on $T^*Q$ or $\RR$ and its pull-back to $T^*Q \times \RR$.
In local coordinates, we have
$$
\eta = dz - p_i \, dq^i.
$$
So, $\eta$ is a contact form on $T^*Q \times \RR$ and $(q^i, p_i, z)$ are Darboux coordinates.
Therefore, we can obtain a Hamiltonian vector field $X_H$ which locally takes the same form that above.

\subsubsection{The Legendre transformation}

Given a Lagrangian function $L : TQ \times \mathbb{R}\to \RR$ we can define the Legendre transformation 
$$
FL : TQ \times \mathbb{R}\to T^*Q \times \RR,
$$
given by
$$
FL (q^i, \dot{q}^i, z) = (q^i, \hat{p}_i, z),
$$
where
$$
\hat{p}_ i = \frac{\partial L}{\partial \dot{q}^i}.
$$
A direct computation shows that
$$
FL ^* \eta = \eta_L,
$$
and then we have
$$
T(FL)(\bar{\xi}_L) = X_H,
$$
and consequently the generalized or contact Euler-Lagrange equations are transformed into the contact Hamilton equations.

\section{Contact manifolds as Jacobi structures}

    Let  $(M,\eta)$ be a $2n+1$ dimensional contact manifold and $\eta\in \Omega^{1}(M)$. We define the Reeb vector field $\Reeb$ and the vector bundle isomorphism $\bar{\flat}$ as in Section~\ref{sec:contact_hamiltonian_systems}. $\sharp$ will denote the inverse of $\bar{\flat}$.

    Given a contact $2n+1$ dimensional manifold $(M, \eta)$, we can consider the following distributions on $M$, that we will call \emph{vertical} and \emph{horizontal} distribution, respectively:
    \begin{eqnarray*}
        \HorD &= &\ker \eta, \\
        \VertD &= &\ker d \eta.
    \end{eqnarray*}

We have a Whitney sum decomposition

    $$
        TM = \HorD \oplus \VertD,
  $$
   and, at each point $x\in M$:
    $$
        T_x M = \HorD_x \oplus \VertD_x.
    $$
    We will denote by $\pHor$ and $\pVert$ the projections onto these subspaces. 
We notice that $\dim \HorD = 2 n$ and $\dim \VertD = 1$, and that $(d\eta)_{|_{\HorD}}$ is non-degenerate.  Moreover, $\VertD$ is generated by $\Reeb$.

\begin{definition}
  \begin{enumerate}
\item A diffeomorphism between two contact manifolds $F:(M,\eta)\to (N, \xi)$ is a \emph{contactomorphism} if
    $$
        F^*\xi = \eta.
   $$

  \item A diffeomorphism $F:(M,\eta)\to (N, \xi)$ is a \emph{conformal contactomorphism} if there exist a nowhere zero function $f\in C^\infty(M)$ such that 
   $$
            F^*\xi = f \eta.
   $$
    
\item A vector field $X \in \mathfrak{X} M$ is an \emph{infinitesimal contactomorphism} (respectively \emph{infinitesimal conformal contactomorphism}) if its flow $\phi_t$ consists of contactomorphisms (resp. \emph{conformal contactomorphisms}).
\end{enumerate}
\end{definition}

Therefore, we have

\begin{proposition}

\begin{enumerate}
  \item A vector field $X$ is an infinitesimal contactomorphism if and only if
$$
        {\mathcal{L}}_{X} \eta = 0.
  $$

\item    $X$ is an infinitesimal conformal contactomorphism if and only if there exists $g \in C^{\infty}(M)$ such that
$$
        {\mathcal{L}}_{X} \eta = g \eta.
$$

    In this case, we say that $(g, X)$ is an \emph{infinitesimal conformal contactomorphism}.
\end{enumerate}
\end{proposition}

    Let $(M, \eta)$ be a $(2n+1)$-dimensional contact manifold. Around any point $x \in M$ there are coordinates $(q^1,\ldots,q^n,p_1,\ldots,p_n,z)$ such that:
$$
        \eta = d z -  p_i d q^i.
$$

    In these coordinates we have
   $$
        d\eta =  d q^i \wedge d p_i \; , \;
        \Reeb = \frac{\partial}{\partial z},
$$
    and
$$
            \VertD = \langle  \frac{\partial}{\partial z} \rangle  \; , \;
            \HorD  = \langle A_i, B^i\rangle 
$$
  where
\begin{eqnarray*}
        A_i &= &\frac{\partial}{\partial q^i} - p_i \frac{\partial}{\partial z},\\
        B^i &=& \frac{\partial}{\partial p_i}.
    \end{eqnarray*}

 $\{A_1, B^1, \ldots, A_n, B^n, \Reeb\}$ and $\{d q^1, d p_1, \ldots, d q^n, d p_n, \eta\}$ are dual basis.

   We also have
$$
        [A_i, B^i ] = -\Reeb
$$

\begin{definition}
    A Jacobi manifold~\cite{Kirillov1976,Lichnerowicz1978} is a triple $(M,\Lambda,E)$, where $\Lambda$ is a bivector field (a skew-symmetric contravariant 2-tensor field) and $E \in \mathfrak{X} (M)$ is a vector field, so that the following identities are satisfied:
   $$
        [\Lambda,\Lambda] = 2 E \wedge \Lambda \; , \;
        {\mathcal{L}}_{E} \Lambda = [E,\Lambda] = 0,
$$
    where $[\cdot,\cdot ]$ is the Schouten–Nijenhuis bracket.
\end{definition}

    Given a Jacobi manifold $(M,\Lambda,E)$, we define the \emph{Jacobi bracket}:
    
        \begin{eqnarray*}
           \{\cdot, \cdot\} : C^\infty(M) \times C^{\infty}(M) & \mapsto \RR , \\
            (f,g) &\mapsto \{f,g\},
        \end{eqnarray*}
    
   \noindent where
    $$
        \{f,g\} = \Lambda(d f, dg) + f E(g) - g E (f).
    $$

This bracket is bilinear, antisymmetric, and satisfies the Jacobi identity. Furthermore it fulfills the weak Leibniz rule:
$$
        supp(\{f,g\}) \subseteq supp (f) \cap supp (g).
 $$
    That is, $(C^\infty(M), \{\cdot,\cdot\})$ is a local Lie algebra in the sense of Kirillov. 

Conversely, given a local Lie algebra $(C^\infty(M), \{\cdot,\cdot\})$, we can find a Jacobi structure on $M$ such that the Jacobi bracket coincides with the algebra bracket.

\begin{remark}{\rm The weak Leibniz rule is equivalent to this identity:

$$
\{f, gh\} = g \{f, h\} + h \{f, g\} + gh E(h)
$$
}
\end{remark}

Given a contact manifold $(M,\eta)$ we can define a Jacobi structure $(M, \Lambda, E)$ by 
$$
    \Lambda(\alpha,\beta) = - d \eta (\sharp\alpha, \sharp\beta), \quad
    E = - \Reeb,
$$
where $\sharp = \bar{\flat}^{-1}$. 

\begin{example}{\rm ({\bf Examples of Jacobi manifolds})

One important particular case of Jacobi manifolds are Poisson manifolds (when $E=0$). The corresponding
Poisson bracket satisfies the following Leibniz rule
$$
        \{f, gh\} = \{f, g\} h + g \{f, h\}.
$$

Examples of Poisson manifolds are symplectic and cosymplectic manifolds, as we show in the following lines.

    Let $(M, \Omega, \eta)$ be a cosymplectic manifold and $\flat: T M \to  T^* M$ be the vector bundle isomorphism defined in Section~\ref{sec:cosymplectic_hamiltonian_systems}

    If we denote its inverse by $\sharp = \flat^{-1}$, then
    $$
        \Lambda(\alpha, \beta) = 
        \Omega(\sharp\alpha,\sharp\beta),
    $$
    is a Poisson tensor on $M$.

An almost symplectic manifold
is said to be {\it locally conformally symplectic} if for each point $x\in M$
there is an open neighborhood $U$ such that $d(e^{\sigma}\Omega)=0,$ for $\sigma:U\rightarrow \mathbb{R}$, so $(U,e^{\sigma}\Omega)$
is a symplectic manifold. If $U=M$, then it is said to be globally conformally symplectic.

One can see that these local 1-forms $d\sigma$ defines a closed 1-form $\theta$ such that
\begin{equation*}
 d\Omega=\theta \wedge \Omega.
\end{equation*}
The one-form $\theta$ is called the {\it Lee one-form}. Locally conformally symplectic manifolds (L.C.S.) with Lee form $\theta=0$ are symplectic manifolds.
We define a bivector $\Lambda$ on $M$ and a vector field $E$ given by
\begin{equation*}
 \Lambda(\alpha,\beta)=\Omega(\flat^{-1}(\alpha),\flat^{-1}(\beta))=\Omega(\sharp(\alpha),\sharp(\beta)),\quad E=\flat^{-1}(\theta)
\end{equation*}
with $\alpha,\beta\in \Omega^{1}(M)$ and $\flat: \mathfrak{X}(M)\rightarrow \Omega^{1}(M)$ is the isomorphism of $C^{\infty}(M)$
modules defined by $\flat(X)=\iota_X\Omega$. Here $\sharp=\flat^{-1}$. In this case, we also have $\sharp_{\Lambda}=\sharp$. The vector field $E$ satisfies $\iota_{E}\theta=0$ and $\mathcal{L}_{E}\Omega=0, \mathcal{L}_{E}\theta=0$.
Then, $(M,\Lambda,E)$ is an even dimensional Jacobi manifold. 

}

\end{example}

    Let $(M,\Lambda,E)$ be a Jacobi manifold. We define the following morphism of vector bundles:
   
        \begin{eqnarray*}
            \lsharp: T M^* &\to     T M\\
            \alpha &\mapsto \Lambda(\alpha, \cdot ),
        \end{eqnarray*}
    which also induces a morphism of $C^{\infty}(M)$-modules between 1-forms and vector fields.

    In the case of a contact manifold, this is given by
 $$
        \lsharp \alpha =
   \sharp \alpha - \alpha(\Reeb) \Reeb,
$$
since
$$
\eta (\sharp_\Lambda \alpha) = \alpha (\Reeb)
$$
for any 1-form $\alpha$.

For a contact manifold, $\lsharp$ is not an isomorphism. In fact, $\ker\lsharp= \langle \eta\rangle $ and ${\rm Im} \; \lsharp = \HorD$.

Vector fields associated with functions $f$ on the algebra of smooth functions $C^{\infty}(M)$ are defined as
\begin{equation*}
 X_f=\sharp_{\Lambda}(df)+fE,
\end{equation*}

The {\it characteristic distribution} $\mathcal{C}$ of $(M,\Lambda,E)$ is generated by the
values of all the vector fields $X_f$. 
This characteristic distribution $\mathcal{C}$ is defined in terms of $\Lambda$ and $E$ as follows
\begin{equation*}
 \mathcal{C}_p=\sharp_{\Lambda_p}(T^{*}_pM)+\langle E_p\rangle ,\quad \forall p\in M
\end{equation*}
where $\sharp_p:T_p^{*}M\rightarrow T_pM$ is the restriction of $\sharp_{\Lambda}$ to $T^{*}_pM$ for every $p\in M$.
Then, $\mathcal{C}_p=\mathcal{C}\cap T_{p}M$ is the vector subspace of $T_pM$ generated by
$E_p$ and the image of the linear mapping $\sharp_p$. 

The distribution is said to be {\it transitive} if the characteristic distribution
is the whole tangent bundle $TM$. The local structure of Jacobi manifolds is described by the following theorem~\cite{Weinstein1983,Sussmann1973}.

\begin{theorem}
The characteristic distribution of a Jacobi manifold $(M,\Lambda,E)$ is completely integrable in the sense of Stefan--Sussmann, thus
$M$ defines a foliation whose leaves are not necessarily of the same dimension, and it is called the {\it characteristic foliation}. Each leaf has a unique
transitive Jacobi structure such that its canonical injection into $M$ is a Jacobi map (that is, it preserves the Jacobi brackets).
Each can be
\begin{enumerate}
 \item A locally conformally symplectic (or a symplectic) manifold if the dimension is even.
\item A manifold equipped with a contact one-form if its dimension is odd.
\end{enumerate}

\end{theorem}

\section{Submanifolds and the Coisotropic Reduction Theorem}

\subsection{Submanifolds}

As in the case of symplectic manifolds, we can consider several interesting types of submanifolds of a contact manifold $(M,\eta)$. To define them, we will use the following notion of \emph{complement} for contact structures:

    Let $(M,\eta)$ be a contact manifold and $x\in M$. Let $\Delta_x\subset T_x M$ be a linear subspace. We define the \emph{contact complement} of $\Delta_x$
    $$
        {\Delta_x}^{\perp_{\Lambda}} = \lsharp({\Delta_x}^o),
    $$
    where 
    ${\Delta_x}^o = \{\alpha_{x}\in T_x^*M \mid \alpha_x(\Delta_x)=0\}$ is the annihilator.

    We extend this definition for distributions $\Delta\subseteq TM$ by taking the complement pointwise in each tangent space.

\bigskip

 \begin{definition}
 Let $N\subseteq M$ be a submanifold. We say that $N$ is:
    \begin{itemize}
        \item \emph{Isotropic} if $TN\subseteq {TN}^{\perp_{\Lambda}}$.
        \item \emph{Coisotropic} if $TN\supseteq {TN}^{\perp_{\Lambda}}$.
        \item \emph{Legendrian} if $TN= {TN}^{\perp_{\Lambda}}$.
    \end{itemize}
\end{definition}

The coisotropic condition can be written in local coordinates as follows.
  
    Let $N\subseteq M$ be a $k$-dimensional manifold given locally by the zero set of functions 
$\phi_a:U\to \mathbb{R}$, with $a\in \{1, \ldots, k\}$. 

    We have that
   $$
        {TN}^{\perp_{\lambda}} = \langle Z_a \; | \; a=1, \dots, k \rangle 
   $$
    where
 $$
        Z_a = \sharp_\Lambda (d \phi_a)
$$

Therefore, $N$ is coisotropic if and only if, $Z_a(\phi_b)=0$ for all  $a,b$. 

Notice that
\begin{equation}\label{za}
Z_a = (\frac{\partial \phi_a}{\partial q^i} + p_i \frac{\partial \phi_a}{\partial z}) \frac{\partial}{\partial p_i}
+ \frac{\partial \phi_a}{\partial p_i} (\frac{\partial}{\partial q^i} - p_i \frac{\partial}{\partial z}).
\end{equation}

According to (\ref{za}), we conclude that $N$ is coisotropic if and only if

\begin{equation}\label{za2}
(\frac{\partial \phi_a}{\partial q^i} + p_i \frac{\partial \phi_a}{\partial z}) \frac{\partial \phi_b}{\partial p_i}
+ \frac{\partial \phi_a}{\partial p_i} (\frac{\partial \phi_b}{\partial q^i} - p_i \frac{\partial \phi_b}{\partial z}) = 0.
\end{equation}

Using the above results, one can easily prove the following characterization of a Legendrian submanifold.

\begin{proposition}
Let $(M, \eta)$ be a contact manifold of dimension $2n+1$. A submanifold $N$ of $M$ is Legendrian if and only if it is a maximal integral manifold of $\ker \eta$
(and then it has dimension $n$).
\end{proposition}

\subsection{Submanifolds in Jacobi manifolds}

Legendre (or Legendrian) submanifolds are a particular case of a more general definition 
for an arbitrary Jacobi manifold. Indeed, let $(M,\Lambda,E)$ be a Jacobi manifold with characteristic distribution $\mathcal{C}$.

\begin{definition}
 A submanifold $N$ of a Jacobi manifold $(M,\Lambda,E)$ is said to be a {\it Lagrangian-Legendrian~\cite{Ibanez1997} submanifold} if the following equality holds
$$
 TN^{\perp_\Lambda} = \sharp_\Lambda (TN^{\circ})=TN \cap \mathcal{C},
$$
where $TN^{\circ}$ denotes the annihilator of $TN$.

$(M,\Lambda,E)$ is said to be {\it transitive} if its characteristic distribution $\mathcal{C}$ is the whole tangent bundle, and then
the above condition reads as

$$
  TN^{\perp_\Lambda} = \sharp_\Lambda (TN^{\circ})=TN.
$$
 
\end{definition}

If $(M,\Lambda)$ is a Poisson manifold, the Lagrangian-Legendrian submanifold of $M$ will simply be called Lagrangian.
In addition, if the Jacobi manifold is contact, then the Lagrangian-Legendrian submanifolds coincide with the Legendre (or Legendrien) submanifolds.

\subsection{Characterization of the dynamics in terms of Legendre submanifolds}

    Given a smooth function $H$ on a contact manifold $(M,\eta)$, we have the \emph{Hamiltonian vector field} 
    $$
        X_H = \sharp_\Lambda (dH) - H \Reeb,
    $$
    or, equivalently,
    $$
        \flat (X_H) = dH - (\Reeb(H) + H) \eta. 
   $$

    In Darboux coordinates, we have
    $$
        X_H =  \frac{\partial H}{\partial p_i} \frac{\partial}{\partial q^i}
        - (\frac{\partial H}{\partial q^i} + p_i \frac{\partial H}{\partial z}) \frac{\partial}{\partial p_i} +  ( p_i \frac{\partial H}{\partial {p_i }} - H) \frac{\partial}{\partial z} .
    $$

\bigskip

    Assume a \emph{contact Hamiltonian} system given by a triple $(M,\eta,H)$, where $(M,\eta)$ is a contact manifold and $H$ is a smooth real function on $M$.

One can easily shows that

$$
{\mathcal{L}}_{X_H} \, H = - \Reeb(H) H.
$$
which shows that the system does not preserve the energy.

Let $(M, \eta, H)$ a Hamiltonian contact system with Reeb vector field $\Reeb$ and Hamiltonian dynamics $X_H$. Assume that $M$ has dimension $2n+1$.

A direct computation shows that
\begin{eqnarray*}
&& {\mathcal{L}}_{X_H} \, \eta = - \Reeb (H) \eta \\
&& {\mathcal{L}}_{X_H}  \, d \eta = - d(\Reeb(H)) \eta - \Reeb (H) d\eta\\
&& {\mathcal{L}}_{X_H}  \, (\eta \wedge d \eta) = - 2 \Reeb (H) \eta \wedge d\eta\\
&& {\mathcal{L}}_{X_H}  \, (\eta \wedge (d \eta)^2) = - 3 \Reeb (H) \eta \wedge (d\eta)^2,
\end{eqnarray*}
and by induction one can prove that
$$
{\mathcal{L}}_{X_H}  \, (\eta \wedge (d \eta)^n) = - (n+1) \, \Reeb (H) \eta \wedge (d\eta)^n.
$$
This proves that the contact volume is not preserved.

However,
$$
\Omega =  {H}^{-(n+1)}  \eta \wedge (d\eta)^n
$$
is preserved, assuming that $H$ does not vanish at every point.

\bigskip

Next, we will investigate the relationship between Hamiltonian vector fields and Legendrian submanifolds. 

\begin{theorem}[Contactification of the tangent bundle]
    Let $(M,\eta)$ be a contact manifold. Let $\bar{\eta}$ be a one form on $TM \times \mathbb{R}$ such that
    $$
        \bar\eta = {\eta}^C + t {\eta}^V,
   $$
    where $t$ is the usual coordinate on $\mathbb{R}$ and ${\eta}^C$ and ${\eta}^V$ are the complete and vertical lifts of $\eta$ to $TM$.
    Then, $(TM \times \mathbb{R}, \bar \eta)$ is a contact manifold with Reeb vector field $\bar\Reeb = {\Reeb}^V$. 
\end{theorem}

\begin{theorem}
    Let $(M,\eta)$ be a contact manifold, and let $X\in\mathfrak{X}(M)$, $f\in C^\infty (M)$. We denote
   
        \begin{eqnarray*}
            X \times f: M &\to &TM \times \mathbb{R}\\
            p &\mapsto & (X_p, f(p)),
        \end{eqnarray*}

    Then $(f, X)$ is an infinitesimal conformal contactomorphism if and only if $\im(X \times f) \subseteq (TM \times \mathbb{R}, \bar{\eta})$ is a Legendrian submanifold.
\end{theorem}

This result states that the image of vector field $X_H$, suitably included in the contactified tangent bundle, is a Legendrian submanifold. In this sense, Hamiltonian vector fields are particular cases of Legendrian submanifolds.
\begin{theorem}
    Let $(M, \eta, H)$ be a contact Hamiltonian system. Then
$$
        \im (X_H \times (\Reeb (H))) \subseteq ( TM \times \mathbb{R}, \bar\eta)
   $$
    is a Legendrian submanifold.
\end{theorem}

The result follows since
$$
{\mathcal{L}}_{X_H} \, \eta = - \Reeb (H) \eta.
$$

\subsection{Coisotropic reduction}

We will present a result of reduction in the context of contact geometry~\cite{Le2018,deLeon2019a}, which is analogous to the well-known coisotropic reduction in symplectic geometry~\cite{Marsden1974,Marsden1986}.

First we note that the horizontal distribution $(\HorD, d\eta)$ is symplectic. Let be $\Delta\subseteq \HorD$. We denote by $\dhbot$ the symplectic orthogonal component

$$
 {\Delta}^{\perp_{d\eta}} =\{v \in TM \; | \; d\eta(v,\Delta)=0\},
$$

We remark that $\Reeb\in  {\Delta}^{\perp_{d\eta}}$ for any distribution $\Delta$. There is a simple relationship between both notions of orthogonal complement:

    Let $\Delta \subseteq TM$ be a distribution. Then
   $$
        {\Delta}^{\perp_{\Lambda}} =  {\Delta}^{\perp_{d\eta}} \cap \HorD.
    $$

We have the following possibilities regarding the relative position of a distribution $\Delta$ in a contact manifold and the vertical and horizontal distributions

\begin{definition}
    Let $\Delta\subseteq TM$ be a distribution of rank $k$. We say that a point $x\in M$ is
    \begin{enumerate}
        \item \emph{Horizontal} if $\Delta_x = \Delta_x \cap \HorD_x$.
        \item \emph{Vertical} if $\Delta_x = (\Delta_x \cap \HorD_x) \oplus \langle \Reeb_x\rangle $.
        \item \emph{Oblique} if $\Delta_x = (\Delta_x \cap \HorD_x) \oplus \langle \Reeb_x + v_x\rangle $, con $v_x \in H_x \setminus \Delta_x$.
    \end{enumerate}

    If $x$ is horizontal, then $\dim {\Delta}^{\perp_{\Lambda}} = 2n - k$. Otherwise, $\dim {\Delta}^{\perp_{\Lambda}}= 2n + 1 -k$. 
\end{definition}

    Given a coisotropic submanifold $\iota : N \to M$, we define
 $$
        \begin{aligned}
            \eta_0 &= \iota^* \eta = {\eta}\mid_{TN}\\
            d \eta_0 &= \iota^* (d \eta) = d {(\iota^*\eta)}.
  \end{aligned}
$$

    We call characteristic distribution of $N$ to 
$$
        {TN}^{\perp_{\Lambda}} =
        \ker(\eta_0) \cap \ker(d \eta_0).
 $$

\begin{theorem}[Coisotropic reduction in contact manifolds]

    Let $\iota:N\to M$ be a coisotropic submanifold of the contact manifold $(M,\eta)$. Then ${TN}^{\perp_{\Lambda}}$ is involutive. 
    
    If the quotient $\tilde N=TN/{TN}^{\perp_{\Lambda}}$ is a manifold and $N$ does not have horizontal points, let $\pi: N \to \tilde N$ be the projection. Then there exists a unique 1-form $\tilde \eta$ on $\tilde{N}$ such that $ \eta = \pi^*(\tilde{\eta})$  and $(N, \tilde \eta)$ is a contact manifold. 

Furthermore, if $N$ consists only of vertical points, then $\tilde \Reeb = \pi_*{\Reeb}$ is well defined and is the corresponding Reeb vector field.
\end{theorem}

The following theorem is very related to a similar result in~\cite{Tortorella2018}.

Indeed, this result provides a coisotropic reduction theorem for \emph{regular coisotropic submanifolds}
which coincides with our notion of \emph{coisotropic submanifolds without horizontal points}, but it is used in a slightly different context.

\begin{remark}{\rm There is another, non-equivalent, widespread definition of contact manifold. 
Some authors define contact manifolds $(M,\xi)$ as odd-di\-men\-sional manifolds $M$ with a \emph{contact distribution} $\xi$, that is, a maximally non-integrable codimension~$1$ distribution. By the Frobenius theorem, this means that $\xi$ is given locally as the kernel of a contact form $\eta$. Of course, every contact manifold $(M,\eta)$ is a contact manifold in this sense by taking $\xi = \ker \eta$. Conversely, a contact distribution $\xi$ is globally the kernel of contact form if and only if $\xi$ is co-orientable.}

\end{remark}

\begin{corollary}
 With the notations from previous theorem, assume that $L\subseteq M$ is Legendrian, $N$ does not have horizontal points, and $N$ and $L$ have clean intersection (that is,  $N \cap L$ is a submanifold and $T(N\cap L) = TN \cap TL$). Then $\tilde L = \pi(L)\subseteq \tilde N$ is Legendrian. 
\end{corollary}

\section{Momentum map and contact reduction}

The moment map is well-known in symplectic geometry~\cite{Marsden1974,Abraham1978}. There is a contact~\cite{Albert1989,Loose2001,deLeon2019a} analog that we will describe below. The moment map has been used to introduce some notions of integrability~\cite{Jovanovic2015,Boyer2011}.
\begin{definition}
    Let $(M,\eta)$ be a contact manifold and let $G$ be a Lie group acting on $M$ by contactomorphisms. In analogy to the exact symplectic case, we define the moment map $J: M  \to \mathfrak{g}^*$ such that
    $$
        J(x)(\xi) = - \eta(\xi_M(x)),
  $$
    where $x \in M$, $\xi\in\mathfrak{g}$ and
  $\xi_M$ is the the infinitesimal generator of the action corresponding to $\xi$.
\end{definition}

We have
    $$
        X_{\hat{J}_{\xi}} =  \xi_M,
   $$
where $X_{\hat{J}_{\xi}}$ is the Hamiltonian vector field corresponding to the function $\hat{J}_{\xi}(x) = \langle J(x), \xi\rangle $.

The moment map defined is equivariant under the coadjoint action. That is, we have
 
   $$
        Ad_{g^{-1}}^* \circ J =
        J \circ g,
    $$
  for $g\in G$, $\alpha\in\mathfrak{g}^*$ and $\xi\in\mathfrak{g}$,
where $Ad^*: G \to Aut(\mathfrak{g}^*)$ is the coadjoint representation.

    Let $(M,\eta)$ be a contact manifold on which a Lie group $G$ acts by contactomorphisms. Let $\mu \in \mathfrak{g}^*$ be a regular value of the moment map $J$. Then, for all $x \in J^{-1}(\mu)$ we have
    $$
        T_x(G_\mu x) = T_x (G x) \cap T_x (J^{-1}(\mu)),
 $$
    where $G_\mu = \{g\in G \; | \;  Ad^*_{g^{-1}} \mu = \mu\}$ is the isotropy group of $\mu$ with respect to the coadjoint action.

   We also have 
$$
        T_x(J^{-1}(\mu)) = {T_x(G x)}^{\perp{d\eta}}.
  $$

 In particular, if $G=G_\mu$, then $T_x (G x) \subseteq T_x (J^{-1}(\mu))$ and $T_x (J^{-1}(\mu))$ is coisotropic and consists of vertical points. Furthermore
    $$
        {T_x(J^{-1}(\mu))}^{\perp_{\Lambda}} = T_x(G x).
    $$

    Let $(M,\eta)$ be a contact manifold on which a Lie group $G$ acts freely and properly by contactomorphisms and let $J$ be the momentum map. Let $\mu\in \mathfrak{g}$ be a regular value of $J$ which is a fixed point of $G$ under the coadjoint action. Then, $M_\mu = J^{-1}(\mu)/G$ has a unique contact form $\eta_\mu$ such that
    $$
        \pi_\mu^* \eta_\mu = \iota_\mu^* \eta,
    $$
    where $\pi_\mu : J^{-1}(\mu) \to M_\mu$ is the canonical projection and $\iota_\mu: J^{-1}(\mu) \to M$ is the inclusion.
    
    Also the Reeb vector field $\Reeb$ restricts to $J^{-1}(\mu)$
    and projects onto $M_\mu$ . Its projection, $\Reeb_\mu$ is the the Reeb vector field of $(M_\mu,\eta_\mu)$.

Let $G$ be a group acting by contactomorphisms on $(M, \eta, H)$ such that $H$ is $G$-invariant. Then,  $(M_\mu,\eta_\mu, H_\mu)$ is a Hamiltonian system, where $H_\mu$ is the induced function by $H$ on $M_\mu$ and
$$
{\pi_\mu}_* {X_H}_{|_{J^{-1}(\mu)}}= X_{H_\mu}.
$$

\section{Infinitesimal symmetries and Noether theorem}
\subsection{Motivation}
Noether's theorem is one of the most relevant results relating symmetry groups of a Lagrangian system and conserved quantities of the corresponding Euler-Lagrange equations~\cite{Carinena1992,Carinena1991,Carinena1989,Carinena1988,Sarlet1987,Cantrijn1980,Sarlet1981,Sarlet1983,deLeon1994,deLeon1994a,Cicogna1992,Prince1983,Prince1985,Crampin1983,Carinena1989a,Aldaya1978,Aldaya1980}. In the simplest view, the existence of a \emph{cyclic} coordinate implies the conservation of the corresponding momentum. Indeed, if $L=L(q^i, \dot{q})$ does not depend on the coordinate $q^j$, then, using the Euler-Lagrange equation
\begin{equation}
    \frac{d}{dt}\left(\frac{\partial L}{\partial \dot{q}^j}\right) - \frac{\partial L}{\partial q^j} = 0,
\end{equation}
we deduce that
\begin{equation}
    \dot{p}_j = \frac{\partial L}{\partial \dot{q}^j} = 0.
\end{equation}

Noether theorem can be described on a geometric framework as follows~\cite{deLeon2011}. 

\begin{theorem}[Noether's Theorem]
Let $L$ be a function on the tangent bundle $TQ$ of the configuration manifold $Q$ and $X$ be a vector field on $Q$. 
Denote by $X^V$ and $X^C$ the vertical and complete lifts of $X$ to $TQ$. Then:

$X^C (L) = 0$ if and only if $X^V(L)$ is a conserved quantity.
\end{theorem}

\medskip

In contact Lagrangian dynamics, the generalized Euler-Lagrange equations look as
\begin{equation}
    \frac{d}{dt}\left(\frac{\partial L}{\partial \dot{q}^j}\right) - \frac{\partial L}{\partial q^j} = \frac{\partial L}{\partial \dot{q}^j} \frac{\partial L}{\partial z},
\end{equation}
and if we insist to proceed as in the symplectic case, we would have
$$
\dot{p}_j = \frac{\partial L}{\partial z} 
$$
but, as we can directly compute
$$
\dot{E}_L = \frac{\partial L}{\partial z} p_j
$$
Therefore, if $E_L$ has no zeros, then $\displaystyle{\frac{p_j}{E_L}}$ is a conserved quantity.

\subsection{Symmetries and contact Hamiltonian systems}

Let $(M, \eta, H)$ a contact Hamiltonian system with Reeb vector field ${\Reeb}$.

The Jacobi bracket of two functions $f, g \in C^\infty(M)$ is given by
$$
\{f, g\} = \Lambda (df, dg) - f {\Reeb} (g) + g {\Reeb}(f),
$$
where ($\Lambda, E = - {\Reeb})$ is the associated Jacobi structure to $(M, \eta)$. Let $X_f$ the Hamiltonian vector field defined by a function $f$.

These two lemmas are essential for our purposes:

\begin{lemma}
We have
$$
\{f, g\}  = X_f (g) + g {\Reeb}(f),
$$

\end{lemma}

This implies that 
$$
X_H (f) = \{H, f\} - {\Reeb}(H) \, f,
$$
so that an observable $f$ dissipates at the same rate that the Hamiltonian if and only if $f$ and $H$ commute
(and in that case, $\frac{f}{H}$ is a conserved quantity.

\bigskip

\begin{lemma}
We have
$$
\{f, g\}  = - \eta ([X_f, X_g ] ).
$$
\end{lemma}

\begin{proposition} Let $X$ be a vector field on $M$ such that $\eta(X) = - f$. Then
$$
\{H, f\} = - \eta ([X_H, X ] ) = ({\mathcal{L}}_X \, \eta) (X_H) + X(H).
$$
\end{proposition}
{\bf Proof}:
If $\eta(X) = - f$, then $\eta(X - X_f) = 0$, so that $X-X_f$ is in the kernel of $\eta$.

Since
$$
{\mathcal{L}}_X \, \eta = - {\Reeb}(H) \eta,
$$
we deduce that
$$
({\mathcal{L}}_X \, \eta)(X_f) = ({\mathcal{L}}_X \, \eta)(X).
$$

Therefore
\begin{eqnarray*}
\{H, f\} &=&  - \eta ([X_H, X_f ] ) \\
& = & ({\mathcal{L}}_{X_H} \, \eta) (X_f) - X_H (\eta(X_f)) \\
&=& ({\mathcal{L}}_{X_H} \eta)(X) - X_H(\eta(X)) \\
&=& - \eta ([X_H, X]) .
\end{eqnarray*}

From the second equality, we have
\begin{eqnarray*}
 - \eta ([X_H, X] )  & = & ({\mathcal{L}}_{X} \, \eta) (X_H) - X (\eta(X_H)) \\
&=& ({\mathcal{L}}_{X} \eta)(X_H) + X(H) .
\end{eqnarray*}

The above Proposition suggests us to introduce the following definition.

\begin{definition}A vector field $X$ on $M$ such that
$$
\eta([X_H, X]) = 0.
$$
will be called a {\sl dynamical symmetry} for $(M, \eta, H)$.
\end{definition}

Using the above Lemmas and the previous Proposition, the following result is immediate.

\begin{theorem}
Let $X$ be a vector field on $M$. Then $X$ is a dynamical symmetry for $(M, \eta, H)$ if and only if $\eta(X)$ commutes with $H$.
\end{theorem}

\subsection{Symmetries and contact Lagrangian systems}

Next, we will consider infinitesimal symmetries on the Lagrangian description.
In this case, we will take benefit fom the bundle structure of $TQ \times \RR$.

For a vector field $X = X^i \frac{\partial}{\partial q^i}$ on $Q$, we will denote its vertical and complete lifts to $TQ$ (with the natural extension to $TQ \times \RR$) by
\begin{eqnarray*}
&&X^V = X^i \frac{\partial}{\partial \dot{q}^i}, \\
&&X^C = \frac{\partial}{\partial q^i} + \dot{q}^j \frac{\partial X^i}{\partial q^j} \frac{\partial}{\partial \dot{q}^i}.
\end{eqnarray*}

Next, let $Y$ be a vector field on $Q \times \RR$. If 
$$
Y = Y^i \frac{\partial}{\partial q^i} + {\mathcal{Z}} \frac{\partial}{\partial z},
$$
then its complete lift to $T(Q \times \RR )$ is
\begin{eqnarray*}
Y^C &=& Y^i \frac{\partial}{\partial q^i} + {\mathcal{Z}} \frac{\partial}{\partial z}
+ \dot{q}^j \frac{\partial Y^i}{\partial q^j} \frac{\partial}{\partial \dot{q}^i} \\
&& + \dot{q}^j \frac{\partial {\mathcal{Z}}}{\partial q^j} \frac{\partial}{\partial \dot{z}}
+ \dot{z} \frac{\partial Y^i}{\partial z} \frac{\partial}{\partial \dot{q}^i} + \dot{z} \frac{\partial {\mathcal{Z}}}{\partial z} \frac{\partial}{\partial \dot{z}}.
\end{eqnarray*}
Here $(z, \dot{z})$ are the bundle coordinates in $T\mathbb{R} \cong \mathbb{R}\times \RR$.

Since we are restricted to the submanifold $TQ \times \RR$ of $T(Q \times \RR )$ we consider only such vector fields $Y$ on $Q \times \RR$ such that its complete lift to
$T(Q \times \mathbb{R})$ 
be tangent to $TQ \times \RR$. This just happens when
$$
\frac{\partial {\mathcal{Z}}}{\partial q^i} = 0,
$$
that is, $\mathcal{Z}$ does not depend on the positions $q$. The restriction of such $Y^C$ to $TQ \times \RR$ will be denoted by
$$
\bar{Y}^C = Y^i \frac{\partial}{\partial q^i} + {\mathcal{Z}} \frac{\partial}{\partial z}
+ \dot{q}^j \frac{\partial Y^i}{\partial q^j} \frac{\partial}{\partial \dot{q}^i}.
$$

In such a case, we will denote by $\bar{Y}^V$ the vertical lift of the projection of $Y$ to $Q$, say
$$
\bar{Y}^V =  Y^i \frac{\partial}{\partial \dot{q}^i} ,
$$
which is obviously tangent to $TQ \times \RR$.

Next, we shall consider a contact Lagrangian system given by a Lagrangian $L : TQ \times \mathbb{R}\to \RR$.
The corresponding contact Hamiltonian system is $(TQ \times \RR , \eta_L, E_L)$ with the obvious notations. ${\Reeb}_L$ is the Reeb vector field and
$\xi_L$ the Euler-Lagrange vector field.

\medskip

\begin{definition} 
A vector field $X$ on $Q$ is called an {\sl infinitesimal symmetry} of $L$ if $X^C(L) = 0$.
\end{definition}

\begin{theorem}
A vector field $X$ on $Q$ is an infinitesimal symmetry of $L$ if and only if
the function
$$
f = X^V(L)
$$ 
commutes with the energy, that is,
$$
\xi_L (f) = - {\Reeb}_L (E_L) f = \frac{\partial L}{\partial z} f.
$$
\end{theorem}

Notice that if $X$ is  an infinitesimal symmetry of $L$, then $X^C$ is the Hamiltonian vector field of $X^V(L)$, say
$$
X^C = X_{X^V(L)}.
$$

The above definition can be slightly extended as follows

\begin{definition} Let $Y$ a vector field on $Q \times \RR$ such that $Y^C$ is tangent to 
$TQ \times \RR$. Then $Y$ is called a {\sl generalized infinitesimal symmetry} of $L$ if 
$$
\bar{Y}^C (L) = - {\Reeb}_L (f) L
$$
where 
$$
f = \bar{Y} (L) - {\mathcal{Z}}
$$
and ${\mathcal{Z}}$ is the $z$-component of $Y$.
\end{definition}

\begin{theorem} Let $Y$ be a {\sl generalized infinitesimal symmetry} of $L$. Then
$$
f = \bar{Y}^V (L) - {\mathcal{Z}}
$$
commutes with $E_L$, and, conversely, in that case, $Y$ is a generalized infinitesimal symmetry of $L$.
\end{theorem}

We can consider more types of infinitesimal symmetries.

\smallskip

\begin{definition} 
A vector field $\tilde{Y}$ on $TQ \times \RR$ is called a {\sl Cartan symmetry} 
if 
$$
{\mathcal{L}}_{\tilde{Y}} \, \eta_L = a \, \eta_L + dg; \tilde{Y}(E_L) = a E_L + g {\Reeb}_L (E_L)
$$
for some functions $a, g \in C^\infty(TQ \times \RR )$.

A vector field $Y$ on $Q \times \RR$ such that $Y^C$ is tangent to $TQ \times \RR$
is called a {\sl Noether symmetry} if $\bar{Y}^C$ is a Cartan symmetry.

\end{definition}

\begin{theorem}

(1) If $Y$ is a Noether symmetry such that
$$
\bar{Y} ^C(L) = g^C
$$
then
$$
f = \bar{Y}^V(L) - g^V
$$
commutes with $E_L$.

(2) If $\tilde{Y}$ is a Cartan symmetry such that
$$
{\mathcal{L}}_{\tilde{Y}} \, \eta_L = dg 
$$
then 
$$
\eta_Y(\tilde{Y}) - g
$$ 
commutes with $E_L$.

\end{theorem}

\begin{definition} 
A vector field $Y$ on $Q \times \RR$ such that
$Y^C$ is tangent to $TQ \times \RR$ and $\bar{Y}^C$ is a dynamical symmetry will be called a {\sl Lie symmetry}.
\end{definition}

\begin{theorem}
If $Y$ is a Lie symmetry, then 
$$
- \eta_L(\bar{Y}^C) = \bar{Y}^V(L) - {\mathcal{Z}}
$$
commutes with $E_L$.
\end{theorem}

\section{Hamilton-Jacobi equation}

We consider the extended phase space $T^{*}Q\times \mathbb{R}$, and a Hamiltonian function
$H:T^{*}Q\times \mathbb{R} \rightarrow \mathbb{R}$. 
\[
\xymatrix{ T^{*}Q\times \mathbb{R}
\ar[dd]^{\rho} \ar[ddrr]^{z}\ar@/^2pc/[ddrr]^{H}\\
  &  & &\\
T^{*}Q &  & \mathbb{R}}
\]
Recall that we have local canonical coordinates $\{q^i,p_i,z\}, i=1,\dots,n$ such that
the one-form is $\eta=dz-\rho^{*}\theta_Q$, $\theta_Q$ being the canonical 1-form on $T^*Q$, can be locally expressed as follows
\begin{equation}\label{contactoneform}
 \eta=dz-\sum_{i=1}^n p_idq^i.
\end{equation}
$(T^{*}Q\times \mathbb{R},\eta)$ is a contact manifold with 
Reeb vector field 
$\mathcal{R}=\frac{\partial}{\partial z}.$

To have dynamics, we consider the vector field
\begin{equation}
 X_H=\sharp_{\Lambda}(dH)+H \Reeb.
\end{equation}
In coordinates, it reads
{\begin{footnotesize}
\begin{equation}\label{1hvf}
 X_H= \sum_{i=1}^n\frac{\partial H}{\partial p_i}\frac{\partial}{\partial q^i} -\sum_{i=1}^n\left(p_i\frac{\partial H}{\partial z}+\frac{\partial H}{\partial q^i}\right)\frac{\partial}{\partial p_i} + \sum_{i=1}^n \; \left(p_i\frac{\partial H}{\partial p_i} -H\right)\frac{\partial}{\partial z} .
\end{equation}
\end{footnotesize}}
We also have
\begin{equation*}
 \flat{(X_H)}= dH -(\mathcal{R}(H)+H)\eta,
\end{equation*}
where $\flat$ is the isomorphism previously defined~\eqref{eq:contact_iso} and 
\begin{equation}\label{1exph}
 \eta(X_H)=-H.
\end{equation}
Recall that $(T^{*}Q\times \mathbb{R},\Lambda,\mathbb{R})$ is a Jacobi manifold with $\Lambda$ given in the usual way.
The proposed contact structure provides us with the {\it dissipative Hamilton equations}.

\begin{equation}\label{hamileq}
\left\{\begin{aligned}
 {\dot q}^i&=\frac{\partial H}{\partial p_i},\\
 {\dot p}_i&=-\frac{\partial H}{\partial q^i}-p_i\frac{\partial H}{\partial z},\\
{\dot z}&=p_i\frac{\partial H}{\partial p_i}-H.
 \end{aligned}\right.
 \end{equation}
for all $i=1,\dots,n$.

Consider $\gamma$ a section of $\pi:T^{*}Q\times \mathbb{R} \rightarrow Q\times \mathbb{R}$, i.e., $\pi\circ \gamma=\text{id}_{Q\times \mathbb{R}}$. We can use $\gamma$ to project $X_H$ on $Q\times \mathbb{R}$
just defining a vector field $X_{H}^{\gamma}$ on $Q\times \mathbb{R}$ by
\begin{equation}\label{hjpar}
 X_H^{\gamma}=T{\pi}\circ X_{H}\circ \gamma,
\end{equation}
where $T\pi$ is the tangent map of $\pi$.
The following diagram summarizes the above construction
\[
\xymatrix{ T^{*}Q\times \mathbb{R}
\ar[dd]^{\pi} \ar[rrr]^{X_H}&   & &T(T^{*}Q\times \mathbb{R})\ar[dd]^{T{\pi}}\\
  &  & &\\
Q\times \mathbb{R} \ar@/^2pc/[uu]^{\gamma}\ar[rrr]^{X^{\gamma}_H}&  & & T(Q\times \mathbb{R})}
\]

We can compute $T\gamma (X_H^\gamma)$ and obtain

\begin{equation}\label{congamma}
T\gamma (X_H^\gamma) = \frac{\partial H}{\partial p_i} \frac{\partial}{\partial q^i} +
(\gamma_i \frac{\partial H}{\partial p_i} - H) \frac{\partial}{\partial z}.
\end{equation}

Therefore, from (\ref{1hvf}) and (\ref{congamma}), we have that
$$
X_H \circ \gamma = T\gamma (X_H ^\gamma)
$$
if and only if

\begin{equation}\label{hjlocal}
\frac{\partial H}{\partial q^j} + 
\frac{\partial H}{\partial p_i} \frac{\partial \gamma_j}{\partial q^i} +
\gamma_j \frac{\partial H}{\partial z} + \gamma_i \frac{\partial \gamma_j}{\partial z} \frac{\partial H}{\partial p_i} - H \frac{\partial \gamma_j}{\partial z} = 0.
\end{equation}

Assume now that

\begin{enumerate}

\item $\gamma(Q\times \mathbb{R})$ is a coisotropic submanifold of 
$(T^{*}Q\times \mathbb{R}, \eta)$;

\item $\gamma_z (Q)$ is a Legendrian submanifold of $(T^{*}Q\times \mathbb{R}, \eta)$, for any $z \in \mathbb{R}$,
where $\gamma_z (q) = \gamma (q, z)$.

\item Notice that the above two conditions imply that $\gamma(Q \times \mathbb{R})$ is foliated by
Legendre leaves $\gamma_z(Q)$, $z \in \mathbb{R}$.

\end{enumerate}

We will discuss the consequences of the above conditions.
The submanifold $\gamma(Q \times \mathbb{R})$ is locally defined by the functions
$$
\phi_i = p_i - \gamma_i = 0.
$$
Therefore, the first condition is equivalent to

\begin{equation}\label{coiso}
\frac{\partial \gamma_i}{\partial q^j} - \gamma_j \frac{\partial \gamma_i}{\partial z} - \frac{\partial \gamma_j}{\partial q^i} + 
\gamma_i \frac{\partial \gamma_j}{\partial z} = 0.
\end{equation}
If, in addition, $\gamma_z(Q)$ is Legendre submanifold for any fixed $z \in \mathbb{R}$, then we obtain
\begin{equation}\label{coiso2}
\frac{\partial \gamma_i}{\partial q^j} - \frac{\partial \gamma_j}{\partial q^i} = 0 
\end{equation}
and, using again (\ref{coiso}), we get
\begin{equation}\label{coiso3}
\gamma_j \frac{\partial \gamma_i}{\partial z} - 
\gamma_i \frac{\partial \gamma_j}{\partial z} = 0.
\end{equation}

Under the above conditions (using \ref{coiso2} and \ref{coiso3}), \ref{hjlocal} becomes

\begin{equation}\label{hjlocal2}
\frac{\partial H}{\partial q^j} + 
\frac{\partial H}{\partial p_i} \frac{\partial \gamma_i}{\partial q^j} +
\gamma_j \left( \frac{\partial H}{\partial z} + \frac{\partial H}{\partial p_i} \frac{\partial \gamma_i}{\partial z} \right) - H \frac{\partial \gamma_j}{\partial z} = 0.
\end{equation}

We can write down eq (\ref{hjlocal2}) in a more friendly way. First of all, consider the following functions and 1-forms defined on
$Q \times \mathbb{R}$:

\begin{enumerate}

\item 
$$
\gamma_o =  \frac{\partial H}{\partial z} + \frac{\partial H}{\partial p_i} \frac{\partial \gamma_i}{\partial z} 
$$
\item 
$$
d(H \circ \gamma_z) = (\frac{\partial H}{\partial q^j} + 
\frac{\partial H}{\partial p_i} \frac{\partial \gamma_i}{\partial q^j} ) dq^j
$$
\item
$$
i_{\frac{\partial}{\partial z}} (d(\gamma^* \theta_Q)) = \frac{\partial \gamma_j}{\partial z} dq^j
$$

\end{enumerate}

Therefore, eq (\ref{hjlocal2}) is equivalent to

\begin{equation}\label{hjglobal}
d (H \circ \gamma_z) + \gamma_o (\gamma^* \theta_Q) - (H\circ \gamma)  (i_{\frac{\partial}{\partial z}} (d(\gamma^* \theta_Q))) = 0.
\end{equation}

\begin{theorem}
 Assume that a section $\gamma$ of the projection $T^*Q \times \mathbb{R} \to Q \times \mathbb{R}$
is such that $\gamma(Q\times \mathbb{R})$ is a coisotropic submanifold of 
$(T^{*}Q\times \mathbb{R}, \eta)$, and $\gamma_z (Q)$ is a Legendrian submanifold of $(T^{*}Q\times \mathbb{R}, \eta)$, for any $z \in \mathbb{R}$. 
Then, the vector fields $X_H$ and $X_H^{\gamma}$ are $\gamma$-related if and only if (\ref{hjlocal2}) holds (equivalently, (\ref{hjglobal}) holds). 
\end{theorem}

Equations (\ref{hjlocal2}) are (\ref{hjglobal}) are indistinctly referred as a {\it Hamilton--Jacobi equation with respect to a contact structure}. A section $\gamma$ fulfilling the assumptions of the theorem and the Hamilton-Jacobi equation will
be called a {\it solution} of the Hamilton--Jacobi problem for $H$.

\begin{remark}{\rm
Notice that if $\gamma$ is a solution of the Hamilton--Jacobi problem for $H$, then
$X_H$ is tangent to the coisotropic submanifold $\gamma(Q \times \mathbb{R})$, but
not necessarily to the Legendre submanifolds $\gamma_z(Q)$. This occurs when
$$
X_H(z -z_0) = 0
$$
for any $z_0$, that is, if and only if
$$
H \circ \gamma_{z_0} = \gamma_i \frac{\partial H}{\partial p_i}.
$$
In such a case, we call $\gamma$ an {\it strong solution} of the Hamilton--Jacobi problem.
}
\end{remark}

Next, we shall discuss the notion of complete solutions of the Hamilton--Jacobi problem for a Hamiltonian $H$.

\begin{definition}
 A {\it complete solution} of the Hamilton--Jacobi equation on a contact manifold $(M,\eta)$ 
 is a diffeomorphism $\Phi:Q\times \RR \times \RR^n\rightarrow T^{*}Q\times \RR$ such that for a set of
 parameters $\lambda\in \RR^n, \lambda=(\lambda_1,\dots,\lambda_n)$, the mapping
 
 \begin{equation}
 \begin{array}{ccc}
  \Phi_{\lambda}:Q\times \mathbb{R}& \rightarrow &  T^{*}Q\times \mathbb{R}  \\
  (q^i, z) &\mapsto &  \Phi(q^i,(\Phi_\lambda)_i(q,z),z)
 \end{array}
 \end{equation}
\noindent
is a solution of the Hamilton--Jacobi equation. If, in addition, any $\Phi_\lambda$ is strong,
then the complete solution is called strong.
\end{definition}

We have the following diagram

\[
\xymatrix{ Q\times \mathbb{R}\times \mathbb{R}^n
\ar[dd]^{\alpha} \ar[rrr]^{\Phi}&   & &T^{*}Q\times \mathbb{R}\ar[dd]^{f_i}\ar[lll]^{\Phi^{-1}}\\
  &  & &\\
 \mathbb{R}^n \ar[rrr]^{\pi_i}&  & & \mathbb{R}}
\]
\noindent
where we define functions $f_i$ such that for a point $p\in T^{*}Q\times \mathbb{R}$, it is satisfied
\begin{equation}\label{functions}
 f_i(p)=\pi_i\circ \alpha\circ \Phi^{-1}(p).
\end{equation}
and $\alpha:Q\times \mathbb{R}\times \mathbb{R}^n\rightarrow \mathbb{R}^n$ is the canonical projection.

The first immediate result is that
$$
\hbox{Im} \; \Phi_\lambda = \cap_{i=1}^n \, f_i^{-1}(\lambda_i)),
$$
where $\lambda = (\lambda_1, \cdots, \lambda_n)$. In other words,
$$
\hbox{Im} \; \Phi_\lambda = \{ x \in T^*Q \times \mathbb{R} \; | \; f_i(x) = \lambda_i, i=1, \cdots, n\}.
$$
Therefore, since $X_H$ is tangent to any of the submanifolds $\hbox{Im} \; \Phi_\lambda$, we deduce that
$$
X_H (f_i) = 0.
$$
So, these functions are conserved quantities.

Moreover, we can compute
$$
\{f_i, f_j\} = \Lambda (df_i, df_j) - f_i \mathcal{R}(f_j) + f_j \mathcal{R}(f_i).
$$
But
$$
\Lambda (df_i, df_j) = \sharp_\Lambda( df_i)(f_j) = 0
$$
since $(T \hbox{Im} \Phi_\Lambda)^\perp = \sharp_\Lambda (T \hbox{Im} \Phi_\Lambda))^o = T \hbox{Im} \Phi_\Lambda)$, so 
\begin{equation}\label{involution}
\{f_i, f_j\} = - f_i \mathcal{R}(f_j) + f_j \mathcal{R}(f_i).
\end{equation}

\begin{theorem}
 There exist no linearly independent commuting set of first-integrals in involution \eqref{functions} for a complete strong solution of the Hamilton--Jacobi
 equation on a contact manifold.
\end{theorem}

{\bf Proof:} If all the particular solutions are strong, then the Reeb vector field $\mathcal{R}$ will be
transverse to the Legendre foliation. So, if the brackets $\{f_i, f_j\} $ vanish, then we would obtain
that the functions $f_i$ cannot be linearly independent.

We remark that some notions of non-commutative integrability have been studied~\cite{Grillo2020}.

\section{Singular Lagrangians and Dirac-Jacobi bracket}

\subsection{Precontact systems}

As we know, if the Lagrangian is regular, then Herglotz's equations, (and, therefore, the variational problem) is equivalent to a contact Hamiltonian system. 

However this is not true for general Lagrangians~\cite{deLeon2019a}. In the following, we will provide some tools to deal with singular Lagrangians. For that, we will need to introduce a geometric model that generalizes contact geometry: precontact geometry. 

This geometry plays a similar role than presymplectic geometry~\cite{Gotay1979} for singular symplectic Lagrangian systems.

Let $\eta$ be a $1$-form in an $m$-dimensional manifold $M$. We define the \emph{characteristic distribution} of $\eta$ as
\begin{equation}
    {\mathcal{C}} = \ker \eta \cap \ker d \eta \subseteq TM.
\end{equation}
We say that $\eta$ is \emph{of class} $c$~\cite{Godbillon1969} if ${\mathcal{C}}$ is a distribution of rank $m-c$.

\begin{proposition}
    Let $\eta$ be a one-form on an $m$-dimensional manifold $M$. The following statements are equivalent:
    \begin{enumerate}
        \item The form $\eta$ is of class $2r+1$.
        \item At every point of $M$,
        \begin{equation}
            \eta \wedge {(d \eta)}^r \neq 0, \quad 
            \eta \wedge {(d \eta)}^{r+1} = 0.
        \end{equation}
        \item Around any point of $M$, there exist local \emph{Darboux} coordinates $x^1,\ldots x^r$, $y_1, \ldots y_r$, $z$, $u_1, \ldots u_s$, where $2r+s+1 = m$, such that
        \begin{equation}
            \eta = d z - \sum_{i=1}^r y_i d x^i.
        \end{equation}
    \end{enumerate}

 In that situation we say that $\eta$ is \emph{a precontact form of class $2r+1$}. 
\end{proposition}

In coordinates, the characteristic distribution is given by
$$
    {\mathcal{C}} = span \{ \frac{\partial}{\partial u_a} _{a=1,\ldots,s} \}.
$$

A pair $(M,\eta)$ of a manifold equipped with a precontact form will be called a \emph{precontact manifold}. A triple $(M,\eta,H)$, where $(M,\eta)$ is a precontact manifold and $H\in {\mathcal{C}}^\infty (M)$ is the \emph{Hamiltonian function} will be called a \emph{precontact Hamiltonian system}.

The distribution ${\mathcal{C}}$ is involutive and it gives rise to a foliation of $M$. If the quotient $\pi:M \to M/{\mathcal{C}}$ has a manifold structure, then there is a unique $1$-form $\tilde\eta$ such that 
    $\pi^* \tilde{\eta}=\eta$. From a direct computation, $\tilde{\eta}$ is a contact form on $M/{\mathcal{C}}$. This justifies the name of \emph{precontact form}.

    We define the following morphism of vector bundles over $M$:
    \begin{equation}
        \begin{aligned}
            \bar{\flat}: TM &\to TM^*\\
            v &\mapsto i_{v} d \eta + \eta(v) \eta.
        \end{aligned}
    \end{equation}
    
    The following $2$-tensors are associated to $\bar{\flat}$ and its transpose
    \begin{equation}
        \omega = d \eta + \eta \otimes \eta, \quad \bar\omega = - d \eta + \eta \otimes \eta.
    \end{equation}
    In other words, $\bar{\flat}(v) = \omega(v,\cdot) = \bar\omega(\cdot,v)$. Therefore $\omega(v,w) = \bar\omega(w,v)$.

    A \emph{Reeb vector field} for $(M,\eta)$ is a vector field ${\Reeb}$ on $M$ such that
    \begin{equation}
        i_{\Reeb} d \eta = 0, \quad \eta({\Reeb}) = 1. 
    \end{equation}
    
    We note that there exists Reeb vector fields in every precontact manifold. Indeed we can define local vector fields ${\Reeb} = \frac{\partial}{\partial z}$ in Darboux coordinates and can extend it using partitions of unity.

 \begin{proposition}
        Let $(M,\eta)$ be a precontact manifold. We have  
        \begin{equation}
            {\mathcal{C}} =  \ker \eta \cap \ker d \eta = \ker \bar{\flat} = {(Im \bar{\flat})}^o.
        \end{equation}
  \end{proposition}

\begin{proposition}
        A vector field $X$ is a Reeb vector field for $(M,\eta)$ if and only if $\bar{\flat}(X) = \eta$. That is, the set of Reeb vector fields is ${\Reeb} + \Gamma({\mathcal{C}})$, where ${\Reeb}$ is an arbitrary Reeb vector field and $\Gamma ({\mathcal{C}})$ is the set of vector fields belonging to ${\mathcal{C}}$.
\end{proposition}

For a distribution $\Delta$ on $M$, we define the following notion of complement with respect to $\omega$. Since $\omega$ is neither symmetric nor antisymmetric, we need to distinguish between right and left complements:

 \begin{eqnarray*}
            {\Delta}^\perp  &=&  
              \{X \in TM \mid \omega(Z,X) = \bar{\flat}(Z)(X) = 0,\, \forall Z\in \Delta \}
              = (\bar{\flat}(\Delta))^o,\\
            ^\perp{\Delta} &=& 
              \{X \in TM\mid \omega(X,Z) = 0, \forall Z \in \Delta \}.
      \end{eqnarray*}

    These complements have the following relationship
    \begin{equation}
        ^\perp({\Delta}^\perp) = 
        (^\perp{\Delta})^\perp = \Delta + {\mathcal{C}}.
    \end{equation}

    We remark that these complements interchange sums and intersections, since the annihilator interchanges them and the linear map $\bar{\flat}$ preserves them. Consequently, if $\Delta,\Gamma$ are distributions, we have
    \begin{eqnarray*}
            (\Delta\cap \Gamma)^\perp &=& {\Delta}^\perp + \Gamma^\perp,\\
            (\Delta +  \Gamma)^\perp &=& {\Delta}^\perp \cap {\Gamma}^\perp.
    \end{eqnarray*}

\subsection{The constraint algorithm}

    We aim to solve Hamilton equations on a precontact Hamiltonian system $(M, \eta, H)$. In order to do that, we will introduce an algorithm similar to the one introduced by M.J. Gotay in 1978~\cite{Gotay1979} for presymplectic systems and that was extended by D. Chinea, M. de Leon, and J.C. Marrero to the cosymplectic case~\cite{Chinea1994}.
    
    Let 
$$
\gamma_H = d H - (H + {\Reeb}(H))\eta
$$ 
where ${\Reeb}$ is a Reeb vector field (we will later see that the algorithm is independent on the choice of the Reeb vector field) and consider the equation
    \begin{equation}
        \bar{\flat}(X) = \gamma_H.
    \end{equation}
This equation might not have solution, so we will consider the subset $M_1 \subseteq M_0 = M$ of the points at which a solution exists. That is,
    \begin{equation}
        M_1 = \{p \in M_0 \mid (\gamma_H)_p \in \bar{\flat}(T_p M_0) \}.
    \end{equation}
    We note that this condition is equivalent to the following
    \begin{equation}
        M_1 = \{p \in M_0 \mid \langle (\gamma_H)_p , {TM_0}^\perp \rangle  = 0 \},
    \end{equation}
    since $\bar{\flat}(TM_0) = (\bar{\flat}(TM_0)^o)^o = (TM_0^\perp)^o$.

    If we choose a local basis $\{X_a\}_{j=a}^{k_1}$ of ${TM_0}^\perp$,  we can easily compute the so-called \emph{primary constraint} functions 
$$
\phi^a(p) = \langle d H_p - ({\Reeb}(H) + H) \eta_p ,X_a \rangle 
$$
whose zero set is the manifold $M_1$. We note that ${TM_0}^\perp =(\im \bar{\flat})^o = \ker \bar{\flat} = {\mathcal{C}}$. Hence,
    \begin{equation}
        \langle d H_p - ({\Reeb}(H) + H) \eta_p ,TM_0^\perp \rangle  = 
        \{Z_p(H) = 0 \mid Z_p \in {\mathcal{C}}_p \}.
    \end{equation}
    Therefore, in Darboux coordinates,
    \begin{equation}
        \phi^a = \frac{\partial H}{\partial s^a}.
    \end{equation}
    We note that this implies that ${\Reeb}(H) = \tilde{\Reeb}(H)$ 
along $M_1$ for every Reeb vector field $\tilde{\Reeb}$, since ${\Reeb}_p-\tilde{\Reeb}_p \in {\mathcal{C}}_p$. Consequently, $(\gamma_H)_{|M_1}$ is independent on the choice of the Reeb vector field. Therefore, the election of ${\Reeb}$ doesn't affect the constraints produced by the algorithm.

   Now we can solve Hamilton equations, but, in order to have meaningful dynamics, the solution $X$ should be tangent to the constraint submanifold. Otherwise, a solution of the equations of motion might escape from $M_1$. This tangency condition is equivalent to demand that $\bar{\flat}(X_p) \in \bar{\flat}(TM_p)$ since $\bar{\flat}$ is an isomorphism modulo ${\mathcal{C}}_p$: 
    \begin{equation}
        M_2 = \{p \in M_1 \mid \langle (\gamma_H)_p ,{TM_1}^\perp \rangle  = 0 \},
    \end{equation}
    providing a second constraint submanifold, with its corresponding constraint functions. However, it is not enough. We must again require that the vector field is tangent to the new submanifold. We then get a sequence of submanifolds
    \begin{equation}
    \begin{aligned}
            M_{i+1} &= 
            \{p \in M_i \mid (\gamma_H)_p \in\bar{\flat}(T_p M_i) \} \\ &=
            \{p \in M_i \mid \langle (\gamma_H)_p , (T_p M_i)^\perp  \rangle = 0 \}    
        \end{aligned}
    \end{equation}
    which eventually stabilizes, that is, there exist some $i_f$ such that $M_{i_f} = M_{i_f+1}$. We call this manifold the \emph{final constraint submanifold} and denote it by $M_f$. This submanifold is locally described by the zero set of some constraint functions $\{\phi^j\}_{j=1}^{k_f}$.

\bigskip

Let $L: TQ \times \mathbb R \to \mathbb R$ be a singular Lagrangian function. The objective is twofold: to develop a constraint algorithm in the Lagrangian side, but also the corresponding Hamiltonian counterpart. Of course, we will use the notations introduced in Section 3.
    
    We make the following observation, which is useful for working with precontact systems that come from a Lagrangian. The proof is trivial from the coordinate expression of $d \eta_L$.

\begin{proposition}
        Let $L:TQ\times\mathbb R \to \mathbb R$ be a Lagrangian function. Then, the form $\eta_L$ is precontact of class $2r+1$ if and only if the rank of the Hessian matrix of $L$ with respect to the velocities is $r$ at every point.

\end{proposition}

 Let $E_L = \Delta(L)-L$ be the energy and $\gamma_{E_L}= d E_L - ({\Reeb}(E_L) + E_L)\eta_L$, where $\eta_L$ is a precontact form of class $2r+1$. We remark that $(TQ \times \mathbb R,\eta_L,E_L)$ is a precontact Hamiltonian system. Hence, we can apply the constraint algorithm developed above to the equation 
$$
\bar{\flat}_L (X) = \gamma_{E_L}.
$$

    If we denote $P_1 = TQ \times \RR$, we will obtain a sequence of constraint submanifolds
    \begin{equation}
        \cdots \to P_{i}  \to \cdots \to P_2 \to P_1,
    \end{equation}
    where
    \begin{equation}
        P_{i+1} =
        \{p \in P_i \mid \langle (\gamma_H)_p , {T_p P_i}^\perp \rangle  = 0\}, 
    \end{equation}
    and $P_f$ is the final constraint submanifold. If it has positive dimension, then there would exist a vector field $X$ tangent to $P_f$ that solves the equations of motion along $P_f$.
    
   Of course, this solution will not be unique in general. We would get a new solution by adding a section of ${\mathcal{C}} \cap TP_f$, where ${\mathcal{C}} = \ker \bar{\flat}_L$ is the characteristic distribution.

Now we will develop a Hamiltonian counterpart of this theory.    
We will require the following additional regularity conditions on $L$ to make sure we get a precontact Hamiltonian system which is amenable to the constraint algorithm:

    \begin{definition}
        We say that a contact Lagrangian
        $L$ is \emph{almost regular} if
        \begin{itemize}
            \item $\eta_L$ is precontact.
            \item the Legendre transformation $FL$ is a submersion onto its image.
            \item For every $p\in T^*Q \times \mathbb R$, the fibers ${(FL)}^{-1}(p)$ are connected submanifolds. 
        \end{itemize}
    \end{definition}

 We denote by $M_1$ be the image of $FL$, which will be called the \emph{primary constraint submanifold}. Let ${FL}_1$ denote the restriction of $FL$ to $M_1$, and $g_1: M_1 \to T^*Q \times \mathbb R$ the canonical inclusion.

    The submanifold $M_1$ is equipped with the form $\eta_1 = {g_1}^*(\eta_Q)$, where $\eta_Q$ is the canonical contact form in $T^*Q \times \mathbb R$. By the commutativity of the diagram, we deduce
    \begin{equation}
        {FL}^*(\eta_1)={FL}^*(\eta_Q) = \eta_L.
    \end{equation}

 \begin{proposition}
        Let $L:P_1 = TQ \times \mathbb R  \to \mathbb R$ be an almost regular Lagrangian. Then $\eta_1 = {g_1}^*(\eta_Q)$ is a precontact form of the same class as $\eta_L$.
\end{proposition}

Furthermore, under the almost regularity hypothesis, we can define a Hamiltonian function
$$
H_ 1 : M_ 1 \to \mathbb R,
$$
such that
$$
H_ 1 \circ {FL}_1 = E_L.
$$

 We conclude that if the Lagrangian is almost regular, then $(M_1, \eta_1, H_1)$ is a precontact Hamiltonian system. Thus, we apply the constraint algorithm to the equation 
$$
\bar{\flat}_1(Y) = \gamma_{H_1},
$$
where $\bar{\flat}_1$ is the mapping defined by $\eta_1$. Thus we obtain a sequence of constraint submanifolds
    \begin{equation}
        \cdots \to M_{i} \to \cdots \to M_2 \to M_1
    \end{equation}
    where $M_f$ is the final constraint submanifold.

    We will investigate the connection between the algorithm on the precontact systems $(P_1,\eta_L,E_L)$ and $(M_1, \eta_1, H_1)$.

\begin{figure}
  \centering
    \includegraphics{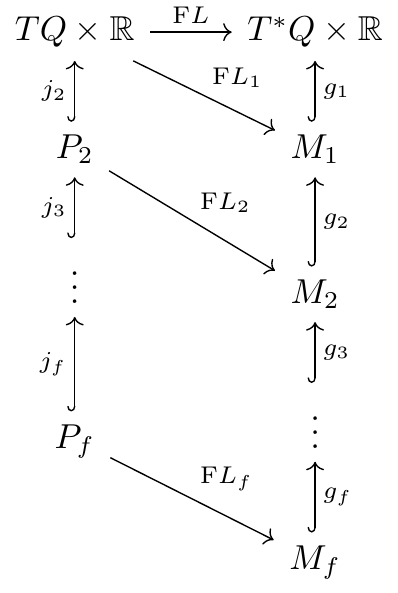}
  \caption{Commutative diagram}
\end{figure}

        Let $L:P\times \mathbb R\to \mathbb R$ be an almost regular Lagrangian, let $(P,\eta_L,E_L)$ be the corresponding precontact system, and let  $(M_1,\eta_1,H_1)$ be its Hamiltonian counterpart. We denote the final constraint submanifolds by $P_f$ and $M_f$, respectively. Then 
        \begin{itemize}
            \item For every $FL$-projectable solution $X$ of the equations of motion along $P_f$, ${FL}_* (X)$ is a solution of Hamilton equations of motion along $M_f$.
            \item  For every solution $Y$ of Hamilton equations of motion along $M_f$, every $X$ such that ${(FL)}_*(X)=Y$ solves the equations of motion along $P_f$.
        \end{itemize}

\subsection{Classification of constraints and the Dirac-Jacobi bracket}

We say that a function $f$ is \emph{first class}~\cite{Ibort1999,Dirac2001} if 
$$
\{f,\phi\}_{M_f} = 0
$$
for any constraint function defining $M_f$.

We say that a function is \emph{second class} if it is not first class.

    We will show that given a family of independent constraints $\phi^\alpha$ defining $M_f$ (by independent, we mean that their differentials are linearly independent) we can extract a maximal subfamily of second class constraints such that the matrix of their Jacobi brackets is non-singular. Modifying the rest of them by taking linear combinations, we get second class constraints that still form an independent family. 

Consider the matrix $(\{\phi^\alpha,\phi^\beta\})_{\alpha,\beta}$. Assume that it has constant rank $k$ in a neighborhood of $M_f$, that is, up to reordering, the first $k$ rows are linearly independent. Denote by $\phi^{a}$ (with latin indices) those functions and $\phi^{\bar{a}}$ (with overlined latin indices) the rest of them. We use greek indices when we want to refer to every constraint. Then the rest of the rows are linear combinations of the first $k$, that is
    $$
        \{\phi^{\bar{a}}, \phi^\beta \} = B^{\bar{a}}_a  \{\phi^a, \phi^\beta \}.
    $$
    Define 
    $$
        {\bar{\phi}}^{\bar{a}} =  {\phi}^{\bar{a}} - B^{\bar{a}}_a \phi^{a}.
    $$
   We can check that these new constraints are first class, so $\phi^a, {\bar{\phi}}^{\bar{a}}$ is a basis of the constraints with the desired properties.

    Now let $C^{ab} = \{\phi^a,\phi^b\}$ and let $C_{ab}$ denote the inverse matrix. We define the \emph{Dirac-Jacobi} bracket such that
$$
        \{f,g\}_{DJ} =  \{f,g\}  - \{f,\phi^a \} C_{ab} \{\phi^b, g\}.
$$

        The Dirac-Jacobi bracket has the following properties:
        \begin{enumerate}
            \item It is a Jacobi bracket which satisfies the generalized Leibniz rule
            $$
                \{fg,h\} = f \{g,h\} + g\{f,h\} + f g {\Reeb}_{DJ}(h),
            $$
            where
            $$
               {\Reeb}_{DJ} = {\Reeb} +  C_{ab} 
               {\Reeb}(\phi^b) (\sharp_{\Lambda}(d \phi^a) + \phi^a {\Reeb}).
            $$
            \item The second class constraints $\phi^a$ are Casimir functions for the Dirac-Jacobi bracket.
            \item For any first class function $F$,
                  \begin{eqnarray*}
                        (\{F,\cdot \} &= &  \{F,\cdot \})_{|{M_f}}, \\
                        ({\Reeb}_{DJ}(F) &= & {\Reeb}(F))_{|{M_f}}.
                  \end{eqnarray*}
            \item The evolution of an observable is given by
            \begin{eqnarray*}      
                    \dot{f} &= &     \{H, f\} - f {\Reeb}_{DJ} (H) +
                    \bar{u}_{\bar{a}} (\{\bar{\phi}^{\bar{a}}, f \}_{DJ} - f {\Reeb}_{DJ} (\bar{\phi}^{\bar{a}})) \\ 
& = &
                    (X_{H}+\bar{u}_{\bar{a}}X_{\bar{\phi}^{\bar{a}}})(f))_{|{M_f}},
            \end{eqnarray*}
            where $H:T^*Q \times \mathbb R \to \mathbb R$ is an arbitrary extension of the Hamiltonian $H_1$.
       \end{enumerate}
        We remark that the motion depends on the multipliers of the first class constraints $\bar{u}_{\bar{a}}$, but it is independent on the multipliers of the second class constraints $u_a$.

\section{Contact nonholonomic dynamics}

Nonholonomic dynamics refers to those mechanical systems that are subject to constraints on the velocities (these constraints could be linear or non-linear).

In the Lagrangian picture, a nonholonomic mechanical system is given by a Lagrangian function $L:TQ \to \RR$ defined on the tangent bundle $TQ$ of the configuration manifold $Q$, with nonholonomic constraints provided by a submanifold $D$ of $TQ$. We assume that $\tau_Q(D)=Q$, where $\tau_Q: TQ \to Q$ is the canonical projection to guarantee that we are in presence of purely kinematic constraints. $D$ could be a linear submanifold (in this case, $D$ can be alternatively described by a regular distribution $\Delta$ on $Q$), or nonlinear.

Even if nonholonomic mechanics is an old subject~\cite{deLeon2019a}, it was in the middle of the nineties that received a decisive boost due to the geometric description by several independent teams: Bloch \emph{et al.}~\cite{Bloch1996}, de León \emph{et al.}~\cite{Ibort1996,deLeon1997a,deLeon1996c,deLeon1996d,deLeon1997b,deLeon1992} and Bates and Śniatycki~\cite{Bates1993}, based on the seminal paper by J.~Koiller in 1992~\cite{Koiller1992}. Another relevant but not so well known work is due to Vershik and Faddeev~\cite{Vershik1972}. A geometrization of nonholonomic mechanics using algebroids is also available~\cite{Grabowski2009}. In~\cite{deLeon2012} the reader can find a historical review on this topic.

Nowadays, nonholonomic mechanics is a very active area of the so-called Geometric Mechanics.

The geometric description of nonholonomic mechanics uses the symplectic machinery. The idea behind is that there exists an unconstrained system as a background and one can recover the nonholonomic dynamics by projecting, for instance, the unconstrained one. Due to their symplectic backstage, the dynamics is conservative (for linear and ideal constraints).

However, there are other kind of nonholonomic systems that do not fit on the above description. On can imagine, for instance, a nonholonomic system subject additionally to Rayleigh dissipation~\cite{Chaplygin1949,MOSHCHUK1987,Neimark1972}. Another source of examples comes from thermodynamics, treated in~\cite{Gay-Balmaz2017,Gay-Balmaz2019} with a variational approach.

Nevertheless, there is a natural geometric description for these systems based on contact geometry. In this section, we will develop a contact version of  Lagrangian systems with nonholonomic constraints.

First we will analyze the Herglotz principle for nonholonomic systems (a sort of d'Alembert principle in comparison with the well-known Hamilton principle).  The reason to develop this subject is to justify the nonholonomic equations proposed in Subsection~\ref{sec:herglotz_constrained}.

Then, in Subsection~\ref{sec:nonholonomic_brackets}, we construct an analog to the symplectic nonholonomic bracket in the contact context. A relevant issue is that this bracket is an almost Jacobi bracket (that is, it does not satisfy the Jacobi identity). This contact nonholonomic bracket transforms the constraints in Casimirs and provides the evolution of observables, as in the unconstrained contact case. In Subsection~\ref{quizalapenultimaseccion} we introduce the notion of almost Jacobi structure proving that the nonholonomic bracket induces, in fact, an almost Jacobi structure. Then, we prove that this structure is a Jacobi structure if, and only if, the constraints are holonomic.

\subsection{Herglotz principle with constraints}\label{sec:herglotz_constrained}
We will consider that the system is restricted to certain (linear) constraints on the velocities modelled by a regular distribution $\Delta$ on the configuration manifold $Q$ of codimension $k$. We will extend the Herglotz principle~\ref{sec:Herglotz_principle} to this case. 

The distribution, $\Delta$ may be locally described in terms of independent linear constraint functions $\{\Phi^a \}_{a=1, \dots , k}$ in the following way
\begin{equation} \label{5}
\Delta=\left\{v\in TQ \mid \Phi^a \left( v \right)=0\right\}.
\end{equation}
Notice that, due to the linearity, the constraint functions $\Phi^{a}$ may be considered as 1-forms $\Phi^{a}: Q \rightarrow T^{*}Q$ on $Q$. Without danger of confusion, we will also denote by $\Phi^{a}$ to the 1-form version of the constraint $\Phi^{a}$ This means that
$$\Phi^a = \Phi^a_i (q) \dot{q}^i.$$
Let $L: TQ \times \mathbb{R} \rightarrow \mathbb{R}$ be the Lagrangian function. One may then define the \emph{Herglotz variational principle with constraints}, that is, we want to find the paths $\xi \in \Omega(q_1,q_1, [a,b])$ satisfying the constraints such that $T_\xi \mathcal{A}(v) = 0$ for all infinitesimal variation $v$ which is tangent to the constraints $\Delta$, where $\mathcal{A}$ is the contact action functional~\eqref{contact_action}. More precisely, we define the set
\begin{equation}
    \Omega(q_1,q_2, [a,b])_{\xi}^{\Delta} = \left\{
        v \in T_\xi \Omega(q_1,q_1, [a,b]) \mid v(t) \in \Delta_{\xi(t)} \text{ for all } t \in [a,b] \right\}.
\end{equation}
Then, $\xi$ satisfies the Herglotz variational principle with constraints if, and only if,
\begin{enumerate}
\item $T_{\xi}\mathcal{A}_{\vert \Omega(q_1,q_1, [a,b])_{\xi}^{\Delta}} = 0.$

\item $\dot{\xi} \left( t \right) \in \Delta_{\xi(t)} \text{ for all } t \in [a,b].$
\end{enumerate}

\begin{definition}
\rm
A \textit{constraint Lagrangian system} is given by a pair $\left( L , \Delta \right)$ where $L : TQ \times \mathbb{R} \rightarrow \RR$ is a regular Lagrangian and $\Delta$ is a regular distribution on $Q$. The constraints are said to be \textit{semiholonomic} if $\Delta$ is involutive and \textit{non-holonomic} otherwise.
\end{definition}

One may easily prove the following characterization of the Herglotz variational principle with constraints.

\begin{theorem}
A path $\xi \in \Omega(q_1,q_2, [a,b])$ satisfies the Herglotz variational principle with constraints if, and only if,
\begin{equation}
    \begin{cases}  
\pdv{L}{q^i} - \dv{}{t} \pdv{L}{\dot{q}^i} + \pdv{L}{\dot{q}^i} \pdv{L}{z}  \in \ann{\Delta}_{\xi(t)}
\\
\dot{\xi} \in \Delta
\end{cases}
\end{equation}
where $\ann{\Delta} = \left\{a \in T^*Q \mid a(u)=0 \text{ for all } u \in \Delta \right\}$ is the annihilator of $\Delta$.
\end{theorem}
Taking into account Eq.~(\ref{5}), we have that $\ann{\Delta}$ is (locally) generated by the one-forms $\Phi^a $. Then $\xi$ satisfies the Herglotz variational principle with constraints if, and only if, it satisfies the following equations
\begin{equation}
    \begin{cases}\label{eq:nonholonomic_herglotz_eqs_coords}
        \dv{}{t} \pdv{L}{\dot{q}^i} - \pdv{L}{q^i} - \pdv{L}{\dot{q}^i} \pdv{L}{z} = \lambda_a\Phi^a_i \\
        \Phi^a (\dot{\xi}(t)) = 0.
    \end{cases}
\end{equation} 
for some Lagrange multipliers $\lambda_i(q^i)$ and where $\Phi^a= \Phi^a_i d q^i$.

From now on, Eqs. (\ref{eq:nonholonomic_herglotz_eqs_coords}) will be called \emph{constraint Herglotz equations}.

We will now present a geometric characterization of the Herglotz equations. In order to do this, we will consider a distribution $\Delta^{l}$ on $TQ\times \mathbb{R}$ induced by $\Delta$ such that its annihilator is given by

\begin{equation}
    \ann{\Delta^{l}} = \left( \tau_{Q} \circ \pr_{TQ \times \mathbb{R}}\right)^{*} \Delta^{0},
\end{equation}
where $\tau_{Q} : TQ \rightarrow Q$ is the canonical projection and $ \pr_{TQ \times \mathbb{R}}: TQ \times \mathbb{R}\rightarrow TQ$ is the projection on the first component. In fact, we may prove that
\begin{equation}
    \Delta^l = S^{*} \left( \ann{T \left(\Delta \times \mathbb{R}\right)}\right).
\end{equation}
Hence, $\ann{\Delta^{l}}$ is generated by the 1-forms on $TQ \times \mathbb{R}$ given by
\begin{equation}\label{eq:delta_l_constraints}
\tilde{\Phi}^{a}=\Phi^{a}_{i}{dq}^{i}.
\end{equation}
Then, we have the following result.

\begin{theorem}\label{18}
Assume that $L$ is regular. Let $X$ be a vector field on $TQ \times \RR$ satisfying the equation
\begin{equation}
    \begin{cases}\label{6}
        \flat_L\left(X\right) - dE_L + \left(E_L + \Reeb_L\left(E_L\right)\right)\eta_L \in \ann{\Delta^{l}} \\
        X_{\vert \Delta \times \RR } \in  \mathfrak{X}\left(\Delta \times \mathbb{R}\right).
    \end{cases},
\end{equation}
Then, 
\begin{itemize}
\item[(1)] $X$ is a SODE on $TQ \times \RR$.

\item[(2)] The integral curves of $X$ are solutions of the constraint Herglotz equations (\ref{eq:nonholonomic_herglotz_eqs_coords}).
\end{itemize}
\end{theorem}

Therefore, Eq. (\ref{6}) provides the correct nonholonomic dynamics in the context of contact geometry. In the case of existence and uniqueness, the particular solution to Eq. (\ref{6}) will be denoted by $\Gamma_{L , \Delta}$. We will now investigate the existence and uniqueness of the solutions.\\

\begin{remark}[The distribution $\Delta^l$]\label{rem:DeltaL}{\rm 
    From the coordinate expression of the constraints $\tilde{\Phi}^a$ defining $\Delta^l$ (Eq.~\eqref{eq:delta_l_constraints}), one can see that $\Reeb_{L}(\tilde{\Phi}^a)=0$, hence $\Delta^l$ is vertical in the sense of Definition 4.}
\end{remark}

\begin{remark}{\rm
Notice that $T \left( \Delta \times \mathbb{R}\right)$ may be considered as a distribution of $TQ \times \RR$ along the submanifold $ \Delta \times \mathbb{R}$. Then, it is easy to show that the annihilator of the distribution $T \left( \Delta \times \mathbb{R}\right)$ is given by $\pr_{\Delta \times \mathbb{R}}^{*} \ann{\left(T \Delta \right)}$ where $\pr_{\Delta \times \mathbb{R}}: \Delta \times \mathbb{R} \rightarrow \Delta$ denotes the projection on the first component. In fact, let $ \left( X , f \right)$ be a vector field on $ TQ \times \mathbb{R}$, that is, for all $\left( v_{q} , z \right) \in T_{q}Q \times \RR$ we have that
\[ X \left( v_{q} , z \right) \in T_{v_{q}}\left( TQ \right) ; \quad f \left( v_{q} , z \right) \in T_{z}\mathbb{R}\cong \RR .\]
Then, for each $\left( v_{q} , z \right) \in \Delta \times \RR$,  $\left( X , f \right) \left( v_{q} , z \right)$ is tangent to $\Delta \times \RR$ at $\left( v_{q} , z \right) $ if, and only if,
\begin{equation}\label{16}
d\Phi^{a}_{\vert v_{q} } \left(  X \left( v_{q} , z \right) \right) = 0 , \ \forall a.
\end{equation}
Denoting $\overline{\Phi}^{a} = \Phi^{a}\circ \pr_{\Delta \times \mathbb{R}}$, we may express Eq. (\ref{16}) as follows
\begin{equation}\label{23}
 Z \left( \overline{\Phi}^{a}\right) = 0, \ \forall a
\end{equation}
where $ Z =\left( X , f \right)$. It is important to notice that, being $\Delta = \left(\Phi^{a}\right)^{-1} \left( 0 \right)$, it satisfies that
\[T_{\left( v_{q} , z \right)} \left( \Delta \times \mathbb{R}\right) = \ker \left( T_{v_{q}} \left(\Phi^{a}\right) \right) \times \RR .\]
}
\end{remark}

\bigskip

Let $\mathcal{S}$ be the distribution on $TQ \times \RR$ defined by $\sharp_{L} \left( \ann{\Delta^{l}}\right)$ where $\sharp_{L} = \flat_{L}^{-1}$.

In order to find a (local) basis of sections of $\mathcal{S}$, we will consider the 1-forms $\tilde{\Phi}^{a} $ generating $\ann{\Delta^{l}}$. For each $a$, $Z_{a}$ will be the local vector field on $TQ \times \RR$ satisfying
\begin{equation}\label{14}
\flat_L \left( Z_{a} \right) = \tilde{\Phi}^{a}.
\end{equation}
Then, $\mathcal{S}$ is  obviously (locally) generated by the vector fields $Z_{a}$ and $\mathcal{S} \subseteq \Delta^{l}$.

By using the proof of the theorem \ref{18} we have that
\begin{equation}
Z_{a}\left( q^{i} \right) = Z_{a} \left( z \right) = 0, \ \ \ \dfrac{\partial^{2}L}{\partial\dot{q}^i \partial \dot{q}^{k}} Z_{a} \left( \dot{q}^{k}\right) = - \Phi^a_i
\end{equation}
Then,
\begin{equation}\label{15}
Z_{a} = -W^{ik} \Phi^{a}_{k}\dfrac{\partial}{\partial \dot{q}^{i}},
\end{equation}
where $\left( W^{ik} \right)$ is the inverse of the Hessian matrix $\left( W_{ik}\right) = \left(\dfrac{\partial^{2}L}{\partial\dot{q}^i \partial \dot{q}^{k}}\right)$.
Notice that, taking into account that $\ann{\Delta^{l}}$ is generated by the 1-forms on $TQ \times \mathbb{R}$ given by $\tilde{\Phi}^{a}=\Phi^{a}_{i}{dq}^{i}$, it follows that
\begin{equation}\label{33}
\mathcal{S} \subseteq \Delta^{l}
\end{equation}

\begin{remark}[The distribution $\mathcal{S}$]\label{rem:S}
    Notice, that, since $\Delta^l$ is vertical (Remark~\ref{rem:DeltaL}), we have $\mathcal{S} = \orthL{(\Delta^l)} = \orth{(\Delta^l)}$ and $\mathcal{S}$ is horizontal. Hence $\eta_{L}(\mathcal{S})=0$.         
\end{remark}

Assume now that there exist two solutions $X$ and $Y$ of Eq. (\ref{6}). Then, by construction we have that $X-Y$ is tangent $T \left( \Delta \times \mathbb{R}\right)$. On the other hand,
\[\flat_{L} \left (X - Y \right) = \flat_{L} \left (X - \Gamma_{L}\right) + \flat_{L} \left (\Gamma_{L} - Y \right) \in \ann{\Delta^{l}}.\]
Then, $X-Y$ is also tangent to $\mathcal{S}$. Thus, we may prove the following result:
\begin{proposition}\label{Yoootraproposiciondelcarajomas}
The uniqueness of solutions of \eqref{6} is equivalent to 
$$\mathcal{S} \cap T \left( \Delta \times \mathbb{R}\right)= \{0\}.$$
\end{proposition}

If the intersection $\mathcal{S} \cap T \left( \Delta \times \mathbb{R}\right)$ were zero, we would be able to ensure the uniqueness of solutions.\\
Let $X =X^{b}Z_{b} $ be a vector field on $\Delta \times \RR$ tangent to $\mathcal{S}$. Hence, by Eq. (\ref{23}), we have that
\[ X^{b} d \overline{\Phi}^{a}\left( Z_{b} \right) = 0.\]
Equivalently,
\[ X^{b}W^{ik}\Phi^{b}_{k}\Phi^{a}_{i} =  0, \ \forall a.\]
Define the (local) matrix $\mathcal{C}$ with coefficient

\begin{equation}\label{lamatrizcdecarajo}
\mathcal{C}_{ab} = - W^{ik}\Phi^{b}_{k}\Phi^{a}_{i} = d\Phi^{b}\left( Z_{a}\right)
\end{equation}
Then, it is easy to prove that (locally) the regularity of $\mathcal{C}$ is equivalent $\mathcal{S} \cap T \left( \Delta \times \mathbb{R}\right)= \{0\}$.\\
One can easily verify that if the Hessian matrix $\left( W_{ik} \right)$ is positive or negative definite this condition is satisfied.\\\\
From now on we will assume that the Hessian matrix $\left( W_{ik} \right)$ is positive (or negative) definite.\\

\begin{remark}
\rm
In general, we may only assume that the matrices $\mathcal{C}$ are regular. However, for applications, in the relevant cases the Hessian matrix $\left( W_{ik} \right)$ is positive definite. In particular, if the Lagrangian $L$ is natural, that is, $L = T+ V\left( q, z \right)$, where $T$ is the \textit{kinetic energy} of a Riemannian metric $g$ on $Q$ and $V$ is a \textit{potential energy}, then the Lagrangian $L$ will be positive definitive.\\

\end{remark}

{Notice that, for each $\left( v_{q} , z \right) \in \Delta \times \RR$ we have that}
\begin{itemize}
\item $\dim \left( S_{ \vert \left( v_{q} , z \right)} \right) = k$

\item $\dim \left(T_{ \left( v_{q} , z \right)} \left( \Delta \times \mathbb{R}\right)\right) = 2n +1 -k $
\end{itemize}
So, the condition of being positive (or negative) definite not only implies that $\mathcal{S} \cap T \left( \Delta \times \mathbb{R}\right)= \{0\}$ but also we have
\begin{equation}\label{41}
\mathcal{S} \oplus T \left( \Delta \times \mathbb{R}\right) = T_{\Delta \times \RR }\left( TQ \times \mathbb{R}\right),
\end{equation}
where $T_{\Delta \times \RR }\left( TQ \times \mathbb{R}\right)$ consists of the tangent vectors of $TQ \times \mathbb{R}$ at points of $\Delta \times \mathbb{R}$.\\
Thus, the uniqueness condition will imply the existence of solutions of \eqref{6}. In fact, we will also be able to obtain the solutions of Eq. (\ref{6}) in a very simple way. In fact, let us consider the two projectors
\begin{subequations}\label{projectors2}
\begin{align}
\mathcal{P}: T_{\Delta \times \RR }\left( TQ \times \mathbb{R}\right) &\rightarrow T \left(  \Delta \times \mathbb{R}\right),\\
\mathcal{Q} :  T_{\Delta \times \RR }\left( TQ \times \mathbb{R}\right) &\rightarrow \mathcal{S}.
\end{align}
\end{subequations}

Consider $X = \mathcal{P} \left( {\Gamma_{L}}_{\vert \Delta \times\RR } \right)$. Then, by definition $X \in  \mathfrak{X} \left(\Delta \times \mathbb{R}\right)$. On the other hand, at the points in $\Delta \times\RR$ we have
\begin{align*}
 &\flat_{L} \left(X \right) - dE_L + \left(E_L + \Reeb_L\left(E_L\right)\right)\eta_L \\ = & \;
\flat_{L} \left(\Gamma_{L} - \mathcal{Q}\left( \Gamma_{L} \right) \right) - dE_L + \left(E_L + \Reeb_L\left(E_L\right)\right)\eta_L\\ =&
 -\flat_{L} \left(\mathcal{Q}\left( \Gamma_{L} \right) \right) \in \ann{\Delta^{l}}
\end{align*}

\noindent{Therefore, by uniqueness, $X_{\vert \Delta \times\RR } = \Gamma_{L , \Delta}$ is a solution of Eq. (\ref{6}).}\\\\
Let us now compute an explicit expression of the solution $\Gamma_{L,\Delta}$. Let $Y$ be a vector field on $TQ \times \RR$. Then, choosing a local basis $\{ \beta_{i}\}$ of $T \left( \Delta \times \mathbb{R}\right)$ we may write the restriction of $Y$ to $\Delta \times \RR$ as follows
\[ Y_{\vert \Delta \times \RR } = Y^{i}\beta_{i} + \lambda^{a}Z_{a}.\]
Then, applying $d\overline{\Phi}^{b}$ we have that
\[ d\overline{\Phi}^{b} \left( Y \right) = \lambda^{a} C_{ba},\]
and we can compute the coefficients $\lambda^{a}$ as follows
\begin{equation}\label{19}
\lambda^{a} = C^{ba}d\overline{\Phi}^{b} \left( Y\right)
\end{equation}
Hence, for all vector field $Y$ on $TQ \times \RR$ restricted to $\Delta \times \RR$
\begin{itemize}
\item $\mathcal{Q} \left( Y_{\vert \Delta \times \RR } \right) = C^{ba}d\overline{\Phi}^{b} \left( Y \right)Z_{a}.$

\item $\mathcal{P} \left( Y_{\vert \Delta \times \RR } \right) = Y_{\vert \Delta \times \RR } - C^{ba}d\overline{\Phi}^{b} \left( Y \right)Z_{a}.$
\end{itemize}

Therefore, we have obtained the explicit expression of the solution $\Gamma_{L, \Delta}$,
\begin{equation}\label{20}
\Gamma_{L,\Delta} = \left(\Gamma_{L}\right)_{\vert \Delta \times \RR } - C^{ba}d\overline{\Phi}^{b} \left( \Gamma_{L} \right)Z_{a}
\end{equation}
\begin{remark}
\rm
From the regularity of the matrices $C$ , we deduce that the projections $\mathcal{P}$ and $\mathcal{Q}$ may be extended to open neighborhoods of $\Delta \times \RR$. Consequently, $\mathcal{P} \left( \Gamma_{L} \right)$ may also be extended to an open neighborhood of $\Delta \times \RR$. However, this extension will not be unique.\\
\end{remark}

Let us recall that the contact Hamiltonian vector fields model the dynamics of dissipative
systems and, contrary to the case of symplectic Hamiltonian systems, the evolution does not preserve the energy, the contact form and the volume, i.e.,
\begin{align*}
    \lieD{\Gamma_L} E_L &= -\Reeb_L (E_L) E_L, \\
    \mathcal{L}_{\Gamma_{L}} \eta_{L} &= -\mathcal{R}_{L} \left( E_{L}\right)\eta_{L}.
\end{align*}
This result may be naturally generalized to the case of non-holonomic constraint by using these projectors.
\begin{proposition}\label{39}
Assume that $L$ is regular. The vector field $\Gamma_{L, \Delta}$ solving the constraint Herglotz equations satisfies that

\begin{subequations}
    \begin{align}
        \mathcal{L}_{\Gamma_{L,\Delta}} \eta_{L} &= 
          - \mathcal{R}_{L} \left( E_{L}\right)\eta_{L}- \lieD{\mathcal{Q}(\Gamma_{L})} \eta_{L},\\
        \mathcal{L}_{\Gamma_{L,\Delta}} \tilde{\eta}_{L} &= 
          - \frac{\lieD{\mathcal{Q}(\Gamma_{L})} \eta_{L}}{H} \\
        \lieD{\Gamma_{L,\Delta}} \Omega_{L} &= -(n+1)\Reeb_L(E_L) \Omega_L  - \eta_L \wedge {d\eta_L}^{(n-1)} \wedge d\mathcal{L}_{\mathcal{Q}(\Gamma_{L})} \eta_{L}\\
        \lieD{\Gamma_{L,\Delta}} \tilde{\Omega}_{L} &= \tilde{\eta}_L \wedge {d\tilde{\eta}_L}^{(n-1)} \wedge d\mathcal{L}_{\mathcal{Q}(\Gamma_{L})} \tilde{\eta}_{L}
    \end{align}
\end{subequations}

where $\tilde{\eta}_L = \eta_L/H$, assuming that $H$ does not vanish,  $\Omega_L = \eta_L \wedge {(d\eta_L)}^{n}$ is the contact volume element and $\tilde{\Omega}_L = \eta_L \wedge {(d\eta_L)}^{n}$.

Furthermore, we have that $\lieD{\mathcal{Q}\left(\Gamma_{L}\right) \eta_{L}} \in  \ann{\Delta^{l}} $.
\end{proposition}

\subsection{Non-holonomic bracket}\label{sec:nonholonomic_brackets}

Consider a regular contact Lagrangian system with Lagrangian $L:TQ \times \mathbb{R}\to \RR$ and constraints
 $\Delta$ satisfying the conditions in Subsection~\ref{sec:herglotz_constrained}. A bracket can be constructed by means of the decomposition~\eqref{41}.

Let us first consider the adjoint operators $\mathcal{P}^*$ and $\mathcal{Q}^*$ of the projections $\mathcal{P}$ and $\mathcal{Q}$, respectively. Obviously, the maps $\mathcal{P}^* :T^{*}_{\Delta \times \RR } \left( TQ \times \mathbb{R}\right) \rightarrow \ann{\mathcal{S}}$ and $\mathcal{Q}^* :T^{*}_{\Delta \times \RR } \left( TQ \times \mathbb{R}\right) \rightarrow \ann{T} \left( \Delta \times \mathbb{R}\right)$ produce a decomposition of $T^{*}_{\Delta \times \RR } \left( TQ \times \mathbb{R}\right)$
\begin{equation}\label{43}
T^{*}_{\Delta \times \RR } \left( TQ \times \mathbb{R}\right) = \ann{\mathcal{S}}\oplus \ann{T} \left( \Delta \times \mathbb{R}\right)
\end{equation}

We may now define, along $\Delta \times \RR$, the following vector and bivector fields:
\begin{align}\label{almostjac24}
    \Reeb_{L,\Delta} &=  \mathcal{P} \left({\Reeb_L}_{\vert \Delta \times \RR }\right),\\
    \Lambda_{L,\Delta} &= \mathcal{P}_* {\Lambda_L}_{\vert \Delta \times \RR },
\end{align}
where $\Lambda_{L}$ is the Jacobi structure associated to the contact form $\eta_{L}$. That is, for $\left( v_{q} , z \right) \in \Delta\times \mathbb{R}\subseteq TQ\times \RR$ and $\alpha,\beta\in T_{\left( v_{q} , z \right)}^* \left(TQ\times \RR \right)$,   
$$\Lambda_{L,\Delta}\left(\alpha,\beta\right) =  \Lambda_{L}\left(\mathcal{P}^*\left(\alpha\right),\mathcal{P}^*\left(\beta\right)\right).$$

{This structure provides the following morphism of vector bundles}
\begin{equation}
  \begin{aligned}
    \sharp_{\Lambda_{L,\Delta}}:  T^{*}_{\Delta \times \RR }\left( TQ \times \mathbb{R}\right) &\to T_{\Delta \times \RR }\left( TQ \times \mathbb{R}\right),\\
    \alpha &\mapsto \Lambda_{L,\Delta}(\alpha, \cdot).
  \end{aligned}
\end{equation}

Hence, we may prove the following result:

\begin{theorem}\label{46}
We have
  \begin{equation}
    \Gamma_{L,\Delta} = \sharp_{\Lambda_{L,\Delta}}(d{E_L}) - {E_L} {\Reeb_{L,\Delta} }.
  \end{equation}
\end{theorem}

Furthermore, we can define the following bracket from functions on $TQ \times \RR$ to functions on $\Delta \times \RR$, which will be called the \emph{nonholonomic bracket}:
\begin{equation}
  \NHBr{f,g} = \Lambda_{L,\Delta}(df, dg) - f \Reeb_{L,\Delta}(g) + g \Reeb_{L,\Delta}(f).
\end{equation}

\begin{theorem}\label{esteteorema23}
  The nonholonomic bracket has the following properties:
  \begin{enumerate}
    \item Any function $g$ on $TQ\times \RR$ that vanishes on $\Delta \times \RR$ is a Casimir, i.e.,
    \[\NHBr{g,f} = 0, \ \forall f \in \mathcal{C}^{\infty} \left( TQ \times \mathbb{R}\right).\]
    \item The bracket provides the evolution of the observables, that is, 
    \begin{equation}
      \Gamma_{L,\Delta}(g) = \NHBr{E_L,g} - g \Reeb_{L,\Delta} (E_L).
    \end{equation}
  \end{enumerate}
\end{theorem}

Notice that, in particular, all the constraint functions $\Phi^{a}$ are Casimir.\\
It is also remarkable that, using the statement \textit{1.} in Theorem~\ref{esteteorema23}, the nonholonomic bracket may be restricted to functions on $\Delta \times \RR$. Thus, from now on, we will refer to the nonholonomic bracket as the restriction of $\NHBr{\cdot , \cdot }$ to functions on $\Delta \times \RR$.

\subsection{Hamiltonian vector fields and integrability conditions}\label{quizalapenultimaseccion}

Until now, we have defined a structure given by a vector field $ \Reeb_{L,\Delta}$ and a bivector field $\Lambda_{L,\Delta} $ which induce the nonholonomic bracket \eqref{nonholbracket243}
\begin{equation}
  \NHBr{f,g} = \Lambda_{L,\Delta}(df, dg) - f \Reeb_{L,\Delta}(g) + g \Reeb_{L,\Delta}(f).
\end{equation}
This structure is quite similar to a Jacobi structure. In fact, we may prove the following result.\\
\begin{proposition}
The nonholonomic bracket endows the space of differentiable functions on $\Delta \times \RR$ with an almost Lie algebra structure~\cite{daSilva1999} which satisfies the generalized Leibniz rule
        \begin{equation}\label{eq:mod_leibniz_rulenonhol}
           \NHBr{f,gh} = g\NHBr{f,h} + h\NHBr{f,g} -  g h \Reeb_{L, \Delta}(h),
        \end{equation}
\end{proposition}
So, as an obvious corollary we have that
\begin{corollary}
The vector field $ \Reeb_{L,\Delta}$ and the bivector field $\Lambda_{L,\Delta} $ induce a Jacobi stucture on $\Delta \times \RR$ if, and only if, the nonholonomic bracket satisfies the Jacobi identity.
\end{corollary}
This result motivates the following definition.
\begin{definition}
Let $M$ be a manifold with a vector field $E$ and a bivector field $\Lambda$. The triple $(M,\Lambda,E)$ is said to be an \textit{almost Jacobi structure} if the pair $(\mathcal{C}^\infty(M),\{\cdot,\cdot\})$ is an almost Lie algebra satisfying the generalized Leibniz rule (\ref{eq:mod_leibniz_rulenonhol}) where the bracket is given by 
\begin{equation}\label{nonholbracket243}
  \{f,g\} = \Lambda(df, dg) + f E(g) - g E(f).
\end{equation}
\end{definition}
With this, the triple $\left( \Delta \times \RR ,  \Lambda_{L,\Delta}, -\Reeb_{L,\Delta}  \right)$ is an almost Jacobi structure. Of course, the study of the intrinsic properties of almost Jacobi structures on general manifolds has a great interest from the mathematical point of view. However, this could distract the reader from the main goal of this paper. So, here we will only focus on the necessary properties for our develoment.\\

Let $H$ be a Hamiltonian function on the contact manifold $\left(TQ \times \RR ,\eta_{L}\right)$. Then, we define the \emph{constrained Hamiltonian vector field} $X_H^{\Delta}$ by the equation
\begin{equation}\label{eq:nonholhamiltonian_vf_contact}
    X_H^{\Delta} = \sharp_{\Lambda_{L,\Delta}} \left(dH \right) - H \Reeb_{L, \Delta} .
\end{equation}
Then, by using \eqref{46} we have that the solution $\Gamma_{L, \Delta}$ of \eqref{6} is a particular case of constrained Hamiltonian vector field. In fact,
$$\Gamma_{L, \Delta} = X_{E_{L}}^{\Delta}.$$

As in the case without constraints, we have many equivalent ways of defining these vector fields.

\begin{proposition}\label{anotherpropomore23}
Let $H$ be a Hamiltonian function on $TQ \times \RR$. The following statements are equivalent:
\begin{itemize}
\item[(i)] $X_{H}^{\Delta}$ is the constrained Hamiltonian vector field of $H$.

\item[(ii)] It satisfies the following equation,
\begin{equation}\label{eq:nonholhamiltonian_vf_contactsecond}
    X_H^{\Delta} = \mathcal{P}\left( \sharp_{L}\left(\mathcal{P}^{*}dH \right)\right) - \left(\Reeb_{L,\Delta} \left(H\right) + H\right) \Reeb_{L,\Delta}.
\end{equation}
\item[(iii)] The following equation holds,
\begin{equation}\label{eq:nonholhamiltonian_vf_contactthird}
    X_{H}^{\Delta} = \mathcal{P}\left( X_{H}\right) - \mathcal{P}\left(   \sharp_{\Lambda_{L}}\left( \mathcal{Q}^{*}dH \right)  \right).
\end{equation}
\end{itemize}
\begin{proof}
Let $g$ a smooth function of $TQ \times \RR$. Then,
\begin{eqnarray*}
X_{H}^{\Delta} \left( g \right) &=&  \{\sharp_{\Lambda_{L,\Delta}} \left(dH \right) \}\left( g \right)- H \Reeb_{L, \Delta} \left( g \right)\\
&=&  \{\mathcal{P}\left( \sharp_{\Lambda_{L}} \left(\mathcal{P}^{*}dH \right)\right) \}\left( g \right) - H \Reeb_{L, \Delta}  \left( g \right)\\
&=&  \{\mathcal{P}\left( \sharp_{L} \left(\mathcal{P}^{*} dH \right) - \mathcal{P}^{*} dH \left( \Reeb_{L}\right) \Reeb_{L} \right)\}\left( g \right) - H \Reeb_{L, \Delta}  \left( g \right)\\
&=&  \{\mathcal{P}\left( \sharp_{L} \left(\mathcal{P}^{*} dH \right) - \Reeb_{L, \Delta}\left( H\right) \Reeb_{L, \Delta} \right) \}\left( g \right) - H \Reeb_{L, \Delta}  \left( g \right)\\
&=& \mathcal{P}\left( \sharp_{L}\left(dH \right)\right) \left( g \right)- \left(\Reeb_{L,\Delta} \left(H\right) + H\right) \Reeb_{L,\Delta} \left( g \right).
\end{eqnarray*}
This proves that $(i)$ is equivalent to $(ii)$. The equivalence between $(i)$ and $(iii)$ follows using the natural decomposition of $\sharp_{\Lambda_{L}}\left( dH \right)$ into $\sharp_{\Lambda_{L}}\left( \mathcal{P}^{*}dH \right)$ and $\sharp_{\Lambda_{L}}\left( \mathcal{Q}^{*}dH \right)$.
\end{proof}
\end{proposition}
Notice that the constrained Hamiltonian vector field $X_{H}^{\Delta}$ is just a vector field along the submanifold $\Delta \times \RR$.
\begin{corollary}
Let $H$ be a Hamiltonian function on $TQ \times \RR$. Then, it satisfies that
\begin{equation}\label{etaproyect34}
\eta_{L}(X^{\Delta}_H) = -H.
\end{equation}

\end{corollary}
As a consequence of this corollary we have that the correspondence $H \mapsto X_{H}^{\Delta}$ is, in fact, an isomorphism of vector spaces. By means of this isomorphism, we may prove the following result.
\begin{proposition}\label{yotramasss}
The nonholonomic bracket $\NHBr{\cdot , \cdot}$ satisfies the Jacobi identity if, and only if, 
$$ \left[ X^{\Delta}_{F}, X^{\Delta}_{G} \right] = X^{\Delta}_{ \NHBr{F,G}},$$
i.e., the correspondence $H \mapsto X^{\Delta}_{H}$ is an isomorphism of Lie algebras.
\end{proposition}
We will now use this result to characterize an integrability condition on the constraint manifold.

\begin{theorem}
The constraint Lagrangian system $\left( L , \Delta \right)$ is semiholonomic if, and only if, the nonholonomic bracket satisfies the Jacobi identity.
\end{theorem}

\begin{proof}
Let $\Phi^{a}$ be the constraint functions. Consider $\tilde{\Phi}^{a}=\Phi^{a}_{i}{dq}^{i}$ the associated $1-$forms generating $\ann{\Delta^{l}}$. Then,
$$\mathcal{P}^{*}\tilde{\Phi}^{a} = \tilde{\Phi}^{a}, \ \forall a.$$
This is a direct consequence of that $\mathcal{P}^{*}dq^{i} = dq^{i}$ for all $i$.
Let us fix $H \in \mathcal{C}^\infty(TQ \times \RR )$. Taking into account that $\mathcal{P}^{*} dH \in \ann{\mathcal{S}}$, we have that
\begin{eqnarray*}
0 & = & \mathcal{P}^{*} dH  \left(  Z_{a} \right)\\
&=&  \mathcal{P}^{*} dH  \left(  \sharp_{L} \left( \tilde{\Phi}^{a}\right) \right)\\
&=& \mathcal{P}^{*} dH  \left(  \sharp_{L} \left( \mathcal{P}^{*}\tilde{\Phi}^{a}\right) \right)
\end{eqnarray*}
Thus, we have that
$$ \tilde{\Phi}^{a} \left(  X_{H}^{\Delta} \right)= 0,$$
i.e., $X_{H}^{\Delta} \in \Delta^{l}$ for all $H \in \mathcal{C}^\infty(TQ \times \RR )$. Let be a (local) basis $\{ X_{b} = X_{b}^{i}\dfrac{\partial}{\partial q^{i}}\}$ of $\Delta$. Then, consider $\Lambda_{b}$ the local functions on $TQ \times \RR$ induced by the $1-$forms $X_{b}^{i}dq^{i}$. Hence, by taking into account that the correspondence $ H \mapsto X_{H}^{\Delta}$ is an isomorphism of vector spaces, we have that the family $\{ X_{\Lambda_{b}}^{\Delta}, X_{z}^{\Delta}\}$ is a (local) basis of $\Delta^{l}$ where $z$ is the natural projection of $TQ \times \RR$ onto $\RR$.\\
So, we have that the distribution $\Delta^{l}$ is involutive.\\
Consider now an arbitrary vector field $X$ on $Q$. Then, there exists a (local) vector field $X^{l}$ on $TQ \times \RR$ which is $\left(\tau_{Q} \circ pr_{TQ \times \RR }\right)-$related with $X$, i.e., the diagram

\begin{large}
\begin{center}
 \begin{tikzcd}[column sep=huge,row sep=huge]
TQ \times \mathbb{R}\arrow[d,"\tau_{Q} \circ pr_{TQ \times \RR }"] \arrow[r, "X^{l}"] &T \left( TQ \times \RR \right) \arrow[d, "T\left(\tau_{Q} \circ pr_{TQ \times \RR } \right)"] \\
 Q  \arrow[r,"X"] &TQ
\end{tikzcd}
\end{center}
\end{large}
is commutative. In fact, let us consider a (local) basis $\{\sigma^{i}\}$ of section of $ \tau_{Q} \circ pr_{TQ \times \RR }$. Then, we may construct $X^{l}$ as follows
$$ X^{l} \left(\lambda_{i} \sigma^{i}\left(q \right)\right) = \lambda_{i} T_{q}\sigma_{i} \left(X \left( q \right) \right),$$
for all $q$ in the domain of the basis. It is finally trivial to check that $X \in \Delta$ if, and only if, any $\left( \tau_{Q} \circ pr_{TQ \times \RR } \right)-$related vector field on $X^{l}$ on $TQ \times \RR$ $X^{l}$ with $X$ is in $\Delta^{l}$. Thus, $\Delta$ is also involutive and, therefore, integrable.

\end{proof}

Therefore, we have proved that the nonholonomic condition of the constraint Lagrangian system $\left(L, \Delta \right)$ may be checked by the Jacobi identity of the nonholonomic brackets.

\section{Other topics}

To avoid an excessive extension of the present survey, we will mention some topics
that we are not including here. We wil, give a brief description of some of them,
and refer to the references where the reader can find more information.

\begin{itemize}

\item {\bf Contact discrete dynamics}

In~\cite{Simoes2020a} the authors  introduce a discrete Herglotz Principle and the corresponding discrete Herglotz Equations for a discrete Lagrangian in the contact setting. This allows us to develop convenient numerical integrators for contact Lagrangian systems that are conformal contactomorphisms by construction. The existence of an exact Lagrangian function is also discussed. Some preliminary results have been discussed in~\cite{Vermeeren2019}, where a construction of variational integrators adapted to contact geometry has been started.

\item{\bf Uniform formalism}

In~\cite{deLeon2020}, the authors develop a unified geometric framework for describing both the Lagrangian and Hamiltonian formalisms of contact autonomous mechanical systems, which is based on the approach of the pioneering work of R. Skinner and R. Rusk~\cite{Skinner1983}. This framework permits to skip the second order differential equation problem, which is obtained as a part of the constraint algorithm (for singular or regular Lagrangians), and is specially useful to describe singular Lagrangian systems. Some examples are also discussed to illustrate the method.

\item {\bf Contact Optimal Control Theory}

In~\cite{deLeon2020a} the authors combine two main topics in mechanics and optimal control theory: contact Hamiltonian systems and Pontryagin Maximum Principle. As an important result, a contact Pontryagin Maximum Principle that permits to deal with optimal control problems with dissipation is developed. Also, the Herglotz optimal control problem is stated, in such a way that generalizes
simultaneously the Herglotz variational principle and an optimal control problem. Some applications to the study of a thermodynamic system are provided.

\item {\bf Existence of invariant measures}

An important topic in dynamical systems is the existence of invariant measures. In~\cite{Bravetti2020}
the authors prove that, under some natural conditions, Hamiltonian systems on a contact manifold $C$
can be split into a Reeb dynamics on an open subset of $C$ and a Liouville dynamics on a submanifold of 
$C$ of codimension 1. Thus, an invariant measure is obtained for the Reeb dynamics, and moreover,a under certain completeness conditions, the existence of an invariant measure for the Liouville dynamics can be characterized using the notion of a symplectic sandwich with contact bread developed in this paper.

\item{\bf Applications to thermodynamics}

In~\cite{Simoes2020b}, the authors, using the Jacobi structure associated with a contact structure, and the so-called evolution vector field, propose a new characterization of isolated thermodynamical systems with friction, a simple but important class of thermodynamical systems which naturally satisfy the first and second laws of thermodynamics, i.e. total energy preservation of isolated systems and non-decreasing total entropy, respectively. In addition, the qualitative dynamics is discussed. Moreover, the discrete gradient methods are applied to numerically integrate the evolution equations for these systems.

\item {\bf Contact higher order mechanics}

In~\cite{deLeon2020c} the authors present a complete theory of higher-order autonomous contact mechanics, which allows us to describe higher-order dynamical systems with dissipation. The essential tools for the theory are the extended higher-order tangent bundles, $T^kQ \times \mathbb{R}$, and its canonical geometric structures. This allow us to state the Lagrangian and Hamiltonian formalisms for these kinds of systems, as well as their variational formulation. In that paper, a unified description that encompasses the Lagrangian and Hamiltonian equations as well as their relationship through the Legendre map; all of them are obtained from the contact dynamical equations and the constraint algorithm that is implemented because, in this formalism, the dynamical systems are always singular. At The theory is applied to some interesting examples.

\item{\bf Classical Field theories with dissipation}

In a series of papers~\cite{Gaset2020,Gaset2020a}, the authors have developed a new geometric framework suitable for dealing with Hamiltonian field theories with dissipation. The geometric is the natural extension of $k$-symplectic structures, so instead to use $k$ copies of the canonical symplectic structure on the cotangent bundle $TM$, the authors consider $k$ copies of the natural contact structure on the extended cotangent bundle $T^*M \times \mathbb{R}$, obtaining the notions of $k$-contact structure and $k$-contact Hamiltonian system. The Lagrangian counterpart is also discussed and related to the Hamiltonian one.

\end{itemize}

\printbibliography

\end{document}